\documentclass[twocolumn, 11pt]{article}
\usepackage[T1]{fontenc}
\usepackage[utf8]{inputenc}

\usepackage{amsmath,amssymb,amsthm}
\usepackage{physics}

\usepackage{graphicx}
\usepackage{subcaption}

\usepackage{comment}
\usepackage{color}

\usepackage[colorlinks=true, linkcolor=blue, urlcolor=blue, citecolor=blue]{hyperref}

\usepackage{titlesec}
\usepackage[
  top=1.8cm,
  bottom=2.0cm,
  left=1.6cm,
  right=1.6cm,
  columnsep=0.6cm
]{geometry}

\renewcommand{\thesection}{\Roman{section}}
\renewcommand{\thesubsection}{\Alph{subsection}}

\titleformat{\section}
  {\normalfont\normalsize\bfseries\centering} 
  {\thesection.}{1em}{\MakeUppercase}

\titleformat{\subsection}
  {\normalfont\normalsize\bfseries\centering}
  {\thesubsection.}{1em}{}

\titlespacing*{\section}{0pt}{2.5em}{0.2em}
\titlespacing*{\subsection}{0pt}{2.5em}{0.4em}

\newenvironment{nalign}{
    \begin{equation}
    \begin{aligned}
}{
    \end{aligned}
    \end{equation}
    \ignorespacesafterend
}

\newtheorem{theorem}{Theorem}
\newtheorem{lemma}{Lemma}

\newtheorem{property}{Property}

\newcommand\blfootnote[1]{%
  \begingroup
  \renewcommand\thefootnote{}\footnote{#1}%
  \addtocounter{footnote}{-1}%
  \endgroup
}
\usepackage[
  backend=biber,
  sorting=none,
  style=numeric-comp,
  giveninits=true,
  maxnames=99,
  doi=true,
  url=false,
  isbn=false,
  eprint=false
]{biblatex}

\DefineBibliographyStrings{english}{
  andothers = {\mkbibemph{et\ al.}}
}

\DeclareNameAlias{author}{given-family}

\renewbibmacro*{volume+number+eid}{%
  \printfield{volume}%
  \iffieldundef{number}
    {}
    {\mkbibparens{\printfield{number}}}%
  \iffieldundef{eid}
    {}
    {\setunit{\addcomma\space}\printfield{eid}}%
}

\AtEveryBibitem{%
  \clearfield{month}%
  \clearfield{day}%
  \clearfield{note}%
  \clearlist{language}%
}

\newbibmacro*{linktail}{%
  \iffieldundef{doi}
    {}
    {\href{https://doi.org/\thefield{doi}}}%
}

\DeclareBibliographyDriver{article}{%
  \printnames{author}%
  \setunit{\addcomma\space}%
  \printfield{title}%
  \setunit{\addcomma\space}%
  \usebibmacro{linktail}{%
    \printfield{journaltitle}%
    \setunit{\addcomma\space}%
    \usebibmacro{volume+number+eid}%
    \setunit{\addcomma\space}%
    \printdate%
  }%
  \finentry
}

\DeclareBibliographyDriver{inproceedings}{%
  \printnames{author}%
  \setunit{\addcomma\space}%
  \printfield{title}%
  \setunit{\addcomma\space}%
  \usebibmacro{linktail}{%
    \printfield{booktitle}%
    \setunit{\addcomma\space}%
    \printlist{publisher}%
    \setunit{\addcomma\space}%
    \usebibmacro{volume+number+eid}%
    \setunit{\addcomma\space}%
    \printdate%
  }%
  \finentry
}

\DeclareBibliographyDriver{incollection}{%
  \printnames{author}%
  \setunit{\addcomma\space}%
  \printfield{title}%
  \setunit{\addcomma\space}%
  \usebibmacro{linktail}{%
    \printfield{booktitle}%
    \setunit{\addcomma\space}%
    \printlist{publisher}%
    \setunit{\addcomma\space}%
    \usebibmacro{volume+number+eid}%
    \setunit{\addcomma\space}%
    \printdate%
  }%
  \finentry
}

\DeclareBibliographyDriver{book}{%
  \printnames{author}%
  \setunit{\addcomma\space}%
    \usebibmacro{linktail}{%
  \printfield{title}%
  \iffieldundef{edition}
    {}
    {\setunit{\addspace}\printfield{edition}\addspace ed.}%
  \setunit{\addcomma\space}%
    \printlist{publisher}%
    \setunit{\addcomma\space}%
    \printdate%
  }%
  \finentry
}

\DeclareBibliographyDriver{thesis}{%
  \printnames{author}%
  \setunit{\addcomma\space}%
  \printfield{title}%
  \setunit{\addcomma\space}%
  \printfield{institution}%
  \setunit{\addcomma\space}%
  \printdate%
  \iffieldundef{doi}
    {}
    {\href{https://doi.org/\thefield{doi}}{}}%
  \finentry
}

\begin{document}
\twocolumn[{
\begin{@twocolumnfalse}

\vspace*{-1.2em}

\begin{center}
{\Large\bfseries Quantum Lattice Boltzmann with Denoising Collision Operators\par}
\vspace{0.6em}

\begin{minipage}{0.88\textwidth} 
\centering

{
Trong Duong$^{1}$, Matthias M\"oller$^{2}$, Norbert Hosters$^{1}$\par
}
\vspace{0.35em}

{\small \itshape
$^{1}$Chair for Computational Analysis of Technical Systems, RWTH Aachen University, \\Schinkelstrasse 2, 52062 Aachen, Germany\par
$^{2}$Delft Institute of Applied Mathematics, Delft University of Technology,\\ Mekelweg 4, 2628 CD Delft, The Netherlands\par
}

\vspace{0.7em}

\begin{center}
\begin{minipage}{\textwidth}
\noindent\textbf{}
The Lattice Boltzmann method (LBM) is a well-established mesoscopic approach for simulating fluid dynamics by evolving particle distribution functions on discrete lattices. While the LBM is highly parallelizable on classical hardware, its translation to quantum algorithms is impeded by the collision process, which is intrinsically nonlinear and irreversible. Several existing quantum formulations implement this process through repeated quantum tomography and state preparation at every timestep, leading to significant overheads. We introduce a quantum LBM based on a denoising-type collision operator that avoids tomography-based updates. The collision dynamics are reformulated as an orthogonal projection onto the linearized manifold of equilibrium distributions around a reference state. This geometric approach filters non-equilibrium components while preserving lattice symmetries and approximating nonlinear terms needed to recover hydrodynamic behavior. A complete pipeline is presented with efficient gate-level realizations, incorporating encoding of distributions, collision, streaming, boundary conditions, and measurement of physical quantities such as hydrodynamic forces. In addition, we outline an approach for implementing projector-based operators deterministically without postselection, paving the way to fully coherent multi-timestep LBM simulations. Numerical experiments for advection-diffusion and flow problems demonstrate that the method reproduces macroscopic behaviors with high accuracy, with performance depending on the choice of reference state.
\end{minipage}
\end{center}

\end{minipage}
\end{center}

\vspace{0.9em}

\end{@twocolumnfalse}
}]

\blfootnote{The implementation of the algorithm described in this paper is publicly available at \url{https://github.com/MyEntangled/denoise_qlbm}; an archived version is available on Zenodo at \url{https://doi.org/10.5281/zenodo.19482608}.}

\section{Introduction}
Simulating flows is a central topic in computational science, with applications across aerodynamics \cite{slotnick_cfd_2014, spalart_role_2016, cary_cfd_2021}, geophysics \cite{eyring_overview_2016, satoh_global_2019}, biological transport \cite{quarteroni_mathematical_2019, marsden_computational_2015, ranno_computational_2025}, and industrial engineering \cite{fletcher_future_2022, sarjito_cfd-based_2021, raynal_cfd_2016}. Among the many numerical approaches developed for this purpose, the Lattice Boltzmann method (LBM) has emerged as a particularly powerful mesoscopic framework \cite{shu_special_2022, petersen_lattice_2021, aidun_lattice-boltzmann_2010}. By evolving discrete particle distribution functions $f_i(\mathbf{x},t)$ on a lattice through alternating collision and streaming steps, LBM combines algorithmic simplicity with scalable parallelism and geometric flexibility. These properties have made LBM competitive with solvers for the Navier-Stokes equations in some regimes. At the same time, the rapid development of quantum computing has motivated new interest in reformulating fluid solvers in a quantum-compatible manner, with the prospect of utilizing quantum parallelism and state superposition to enhance the scalability of flow solvers.

Quantum algorithms for partial differential equations have demonstrated, at least in principle, asymptotic advantages for linear systems \cite{harrow_quantum_2009, childs_quantum_2017, ambainis_variable_2012} and Hamiltonian simulation \cite{lloyd_universal_1996, berry_exponential_2014, berry_simulating_2015, low_optimal_2017} under certain assumptions. However, quantum algorithms for fluid dynamics simulations still face immense difficulties due to intrinsic nonlinearity and complex dynamics of flows \cite{sanavio_quantum_2024, meng_quantum_2023, gaitan_finding_2020}. In this context, LBM provides unique properties: its statistical interpretation, locality, and linear streaming dynamics are compatible with unitary evolutions on subsystems and simple quantum walks, suggesting that LBM may facilitate a quantum approach to fluid simulation.

This observation has motivated the development of the quantum Lattice Boltzmann method (QLBM), in which the classical distribution functions are encoded into quantum states and evolved using quantum circuits. Early conceptual work established the connection between Lattice Gas Automata, the Lattice Boltzmann method, and quantum walks \cite{yepez_quantum_2001, meyer_quantum_1996, yepez_quantum_2001-1, succi_quantum_2015}. Current studies have formulated quantum algorithms for collisionless dynamics \cite{todorova_quantum_2020}, linearized Boltzmann equations \cite{lee_multiple-circuit_2026, itani_quantum_2024, zeng_quantum_2025, kumar_decomposition_2024}, advection-diffusion problems \cite{budinski_quantum_2021, wawrzyniak_linearized_2025}, and simplified hydrodynamic models \cite{zamora_float_2025, kocherla_fully_2024, wang_quantum_2025, georgescu_quantum_2025, georgescu_fully_2026,  lacatus_surrogate_2025}, often highlighting the conceptual elegance and quantum compatibility of the methods. Relevant research also addresses complex geometry \cite{ji_iga-lbm_2025}, software implementation \cite{georgescu_qlbm_2025, shinde_utilizing_2025}, and scalability of QLBM algorithms \cite{jennings_simulating_2026, turro_practical_2025}.

A major line of research has focused on extending QLBM to incorporate the nonlinear collision step. One class of approaches employs Carleman linearization to embed the nonlinear collision term into a higher-dimensional linear system \cite{itani_quantum_2024}, which can then be simulated with linear operations. This strategy can handle moderately nonlinear flows over multiple time steps and has been explored in a range of algorithmic and circuit-level implementations \cite{sanavio_lattice_2024, zamora_quantum_2026}. In practice, however, manipulating the resulting linear system requires deep circuits with complex, multi-qubit operators \cite{sanavio_quantum_2024, sanavio_lattice_2024} as well as post-selection steps with very low success probability \cite{sanavio_carleman-lattice-boltzmann_2025}. Alternatively, other works have considered lattice-gas or surrogate models \cite{wang_quantum_2025, kocherla_fully_2024, georgescu_quantum_2025, lacatus_surrogate_2025} for the collision step, but these approaches do not guarantee an approximation to the Lattice Boltzmann equation.

Despite this progress, fundamental limitations in treating the collision step remain due to its intrinsic nonlinearity and irreversibility. Many existing quantum formulations address this difficulty by resorting to repeated quantum state tomography, mid-circuit measurements, or explicit re-preparation of quantum states at every time step. While these strategies allow one to reintroduce nonlinearity at the algorithmic level, they also incur substantial measurement overhead and loss of coherence, which severely limit scalability. Even approaches based on Carleman linearization face rapidly growing resource requirements or restrictive assumptions on the flow regime. As a result, a fully coherent, measurement-free QLBM capable of reproducing nonlinear hydrodynamic behavior over multiple time steps remains an open problem.

The central gap, therefore, lies in the absence of a collision mechanism that is simultaneously compatible with unitary quantum evolution, avoids tomography-based updates, and retains sufficient nonlinear structure to recover correct macroscopic dynamics. Addressing this gap requires rethinking the role of collision in QLBM beyond a direct quantum translation of its classical counterpart.

In this paper, we propose a quantum Lattice Boltzmann method that replaces tomography-based state updates with a denoising-inspired collision step. Instead of reconstructing the full state, we treat collision as a projection onto a local tangent space of the equilibrium manifold around a reference state. In this way, components that are orthogonal to the tangent space are filtered out. With a suitable reference state, this projection acts like a denoising procedure that drives the initial state toward the corresponding equilibrium state. By alternating applications of this collision operator with the linear streaming step, the QLBM evolution can be carried out over multiple time steps without quantum tomography or state re-preparation (Fig.~\ref{fig:LBM_workflow}).

The remainder of this paper is organized as follows. Section \ref{sec:lbm_background} reviews the standard LBM formulation. Section \ref{sec:qlbm} introduces the geometric interpretation underlying the proposed collision operator and analyzes the approach's symmetry and errors. We also present a complete QLBM pipeline, including quantum encoding of distribution functions, unitary implementations of collision and streaming, treatment of boundary conditions, and measurement of physical observables such as hydrodynamic forces. Furthermore, we outline a route to implementing projector-based collision operators deterministically via double-bracket evolution, enabling efficient multi-timestep QLBM simulations. Section \ref{sec:numerics} presents numerical results and accuracy analyses for flow and advection-diffusion test problems. Finally, Section \ref{sec:conclusion} concludes with a summary and outlook.

\begin{figure}[h!]
    \centering
    \includegraphics[width=\linewidth]{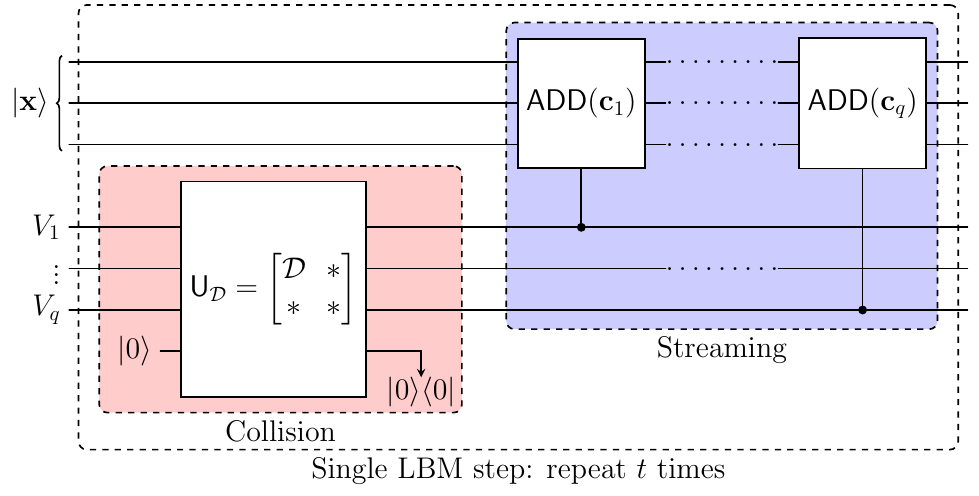}
    \caption{Diagram for standard QLBM simulations over $t$ timesteps without boundary treatment.}
    \label{fig:LBM_workflow}
\end{figure}

\section{Lattice Boltzmann Method}
\label{sec:lbm_background}
The Lattice Boltzmann Method (LBM) simulates fluid dynamics by tracking the evolution of particle distribution functions on a discrete lattice. Unlike traditional Computational Fluid Dynamics (CFD) methods that discretize the macroscopic Navier-Stokes equations, LBM discretizes the mesoscopic Boltzmann kinetic equation. This section reviews the standard LBM with the Bhatnagar\textendash Gross\textendash Krook (BGK) collision model, which serves as the classical benchmark for our quantum formulation.

We consider a $d$-dimensional lattice with $q$ discrete velocity vectors, denoted as the D$d$Q$q$ lattice. At each lattice node $\mathbf{x}$ and time step $t$, the fluid state is described by a set of populations $f_i(\mathbf{x}, t)$ that represent the particle distribution function with respect to discrete velocities $\{\mathbf{c}_i\}_{i=1}^{q}$.

In the standard LBM formulation, with lattice units $\Delta x = 1$ and $\Delta t = 1$, the time evolution of these populations is governed by the BGK Lattice Boltzmann equation,
\begin{equation}
f_i(\mathbf{x} + \mathbf{c}_i, t + 1) = f_i(\mathbf{x}, t) - \frac{1}{\tau} \left[ f_i(\mathbf{x}, t) - f_i^{\text{eq}}(\mathbf{x}, t) \right],
\label{eq:lbm_evolution}
\end{equation}
where $f_i^{\text{eq}}$ is the local equilibrium distribution and $\tau$ is the relaxation time, related to the kinematic viscosity by $\nu = c_s^2 (\tau - 0.5)$, where $c_s$ represents the lattice speed of sound with typical value $c_s = 1/\sqrt{3}$ for common LBM lattices.

The equilibrium distribution $f_i^{\text{eq}}$ is central to the LBM. It is the steady distribution under fixed macroscopic values and depends solely on the local density $\rho$ and velocity $\mathbf{u}$, which are obtained from the moments of the distributions:
\begin{nalign}
\rho(\mathbf{x},t) &= \sum_{i} f_i(\mathbf{x},t), \\
\rho(\mathbf{x},t) \mathbf{u}(\mathbf{x},t) &= \sum_{i} \mathbf{c}_i f_i(\mathbf{x},t)
\label{eq:moments}
\end{nalign}
Using these moments, $f_i^{\text{eq}}$ is calculated via a second-order expansion valid for low Mach numbers, i.e., $|\mathbf{u}|/c_s \ll 1$, by
\begin{equation}f_i^{\text{eq}}(\rho, \mathbf{u}) \approx w_i \rho \left( 1 + \frac{\mathbf{c}_i \cdot \mathbf{u}}{c_s^2} + \frac{(\mathbf{c}_i \cdot \mathbf{u})^2}{2c_s^4} - \frac{\|\mathbf{u}\|^2}{2c_s^2} \right),
\label{eq:equilibrium}
\end{equation}
where $w_i$ is the lattice-specific weight corresponding to the population with velocity $\mathbf{c}_i$. Crucially, the quadratic velocity terms in Eq.~\eqref{eq:equilibrium} are what allow the LBM to recover the nonlinear physics of fluid flow.

Algorithmically, an LBM algorithm based on Eq.~\eqref{eq:lbm_evolution} is executed in three recurring steps: 
\begin{enumerate} 
\item \textit{Collision:} A non-linear process occurs independently at every lattice node simultaneously, in which the populations at each node relax toward a local equilibrium: 
\begin{equation}
\label{eq:lbm_collision}
f_i^{*}(\mathbf{x}, t) = \left(1 - \frac{1}{\tau}\right) f_i(\mathbf{x}, t) + \frac{1}{\tau} f_i^{\text{eq}}(\mathbf{x}, t)
\end{equation} 
The zeroth moment (density $\rho$) and the first moment (momentum $\rho \mathbf{u}$) are invariant during collision.
\item \textit{Streaming:} The post-collision populations advect to neighboring nodes according to their velocity vectors: 
\begin{equation} 
\label{eq:lbm_streaming}
f_i(\mathbf{x} + \mathbf{c}_i, t + 1) = f_i^{*}(\mathbf{x}, t). \end{equation} 
\item \textit{Boundary conditions:} Local populations evolve upon encountering the boundary. The evolution depends on the type of boundary condition.
\end{enumerate}

 It is important to note that while the streaming step consists of simple data shifts, the calculation of $f_i^{\text{eq}}$ in Eq.~\eqref{eq:equilibrium} during the collision step requires nonlinear arithmetic for quadratic terms. This nonlinearity represents the primary barrier for linear quantum evolutions and will be our main focus in this paper.

\section{Quantum Algorithm}
\label{sec:qlbm}
This section focuses on designing an implementation pipeline for LBM simulation on quantum computers. To begin with, we will propose an amplitude-based quantum encoding of the populations. The following subsections will introduce the construction and the gate-level implementation for collision, streaming, and boundary conditions, which occur alternately in a standard LBM simulation.

The particle distribution functions are encoded into a quantum state $\ket{\psi}$ within a composite Hilbert space $\mathcal{H} = \mathcal{H}_D \, \otimes \, \mathcal{H}_V$, comprising a spatial domain register $D$ and a velocity register $V$. The global state is defined as
\begin{equation}
\ket{\psi} = \sum_{\mathbf{x}} \sqrt{\rho(\mathbf{x})} \ket{\mathbf{x}}_{D} \ket{\psi_{\mathbf{x}}}_{V},
\label{eq:encoding}
\end{equation}
where the local state $\ket{\psi_{\mathbf{x}}}$ encapsulates the normalized distribution at lattice site $\mathbf{x}$:
\begin{equation}
\ket{\psi_{\mathbf{x}}} = \frac{1}{\sqrt{\rho(\mathbf{x})}} \sum_{i=1}^{q} \sqrt{f_i(\mathbf{x})} \ket{\mathbf{e}_i}.
\label{eq:local_state}
\end{equation}
Here, $\{ \ket{\mathbf{e}_i} \}_{i=1}^q$ denotes a specific subset of the computational basis vectors in $\mathcal{H}_V$ (where $\dim(\mathcal{H}_V) = 2^q$). These vectors correspond to the one-hot binary strings representing the $q$ discrete velocities. The normalization of $\ket{\psi}$ implies $\sum_{\mathbf{x}} \rho(\mathbf{x}) = 1$, consistent with the conservation of total fluid mass across the lattice.

\subsection{Collision}
Implementing the collision step poses a challenge for quantum LBM algorithms because collision-driven relaxation is nonlinear and irreversible. This part will introduce a quantum operator that approximates the collision step (Eq.~\eqref{eq:lbm_collision}) in the full relaxation regime ($\tau = 1$). The operator is analytically constructible based on an orthogonal projection onto the local expansion of the so-called equilibrium manifold. This approach is called denoising-based collision and is the key contribution of this article. The method is applicable to quantum LBM formulations for both fluid dynamics and advection-diffusion problems.

Applying $\sqrt{1+\varepsilon} \approx 1 + \frac{\varepsilon}{2} - \frac{\varepsilon^2}{8}$ to the equilibrium distribution in Eq.~\eqref{eq:equilibrium} yields
\begin{equation}
\begin{split}
\sqrt{f_i^{\text{eq}}} 
\approx \sqrt{\rho w_i} \left(1 + \frac{\mathbf c_i\cdot \mathbf u}{2c_s^2}
+ \frac{(\mathbf c_i\cdot \mathbf u)^2}{8c_s^4}
- \frac{\|\mathbf u\|^2}{4c_s^2} \right) \\ 
+ \mathcal{O}( \|\mathbf{u}\|^3).
\end{split}
\end{equation}
Therefore, the corresponding amplitude vector $\sqrt{\mathbf{f}^{\text{eq}}} = \left(\sqrt{f_1^{\text{eq}}}, \dots, \sqrt{f_q^{\text{eq}}}\right)$ belongs, up to $\mathcal{O}(\|\mathbf{u}\|^3)$, to the subspace $\mathcal{V} = \mathrm{span}_{\mathbb{R}}\{\mathbf{v}_0, \{\mathbf{v}_{k}\}_{k=1}^d, \{\mathbf{v}_{kl}\}_{1\leq k\leq l \leq d} \}$ of $\mathbb{R}^q$, whose basis vectors are defined entry-wise by
\begin{equation}
\begin{split}
    (\mathbf{v}_0)_i &= \sqrt{w_i}, \\ (\mathbf{v}_k)_i &= \frac{c_{ik}}{c_s}\sqrt{w_i}, \\ (\mathbf{v}_{kl})_i &= \left(\frac{c_{ik} c_{il} - \delta_{kl} \, c_s^2}{\sqrt{1+\delta_{kl}} \, c_s^2}\right)\sqrt{w_i},
\end{split}
\label{eq:hermite-basis}
\end{equation}
where $\delta_{kl}$ is the Kronecker delta. Under the isotropy conditions\footnotemark satisfied by common LBM lattices, these basis vectors form an orthonormal set.\footnotetext{Common LBM lattices satisfy isotropy conditions up to fifth order: $\sum_i w_i = 1$; $\sum_i w_i c_{ik} c_{il} = c_s^2 \delta_{kl}$; $\sum_i w_i c_{ik} c_{il} c_{ip} c_{iq} = c_s^4 (\delta_{kl}\delta_{pq} + \delta_{kp}\delta_{lq} + \delta_{kq}\delta_{lp})$; and all moments of orders $1$, $3$, and $5$ vanish.} The dimension of the subspace is $d_{\mathcal V} = (d+1)(d+2)/2$. The amplitude vector decomposes in this subspace as $\sqrt{\mathbf{f}^{\text{eq}}} = \sqrt{\rho} \mathbf{h}(\mathbf{u}) + \mathcal{O}(\|\mathbf{u}\|^3)$ with
\begin{equation}
    \begin{split}
        \mathbf{h}(\mathbf{u}) = \left( 1 - \frac{\|\mathbf{u}\|^2}{8c_s^2}\right) \mathbf{v}_0 + \sum_{k} \frac{u_k}{2c_s} \mathbf{v}_{k} \\ 
        + \sum_{k\leq l} \frac{u_{k} u_{l}}{4\sqrt{1+\delta_{kl}} \, c_s^2} \mathbf{v}_{kl}
    \end{split}
    \label{eq:velocity_component}
\end{equation}

In the standard LBM simulation workflow, the collision and streaming steps occur in succession. For laminar flows, we interpret the streaming step as a perturbation that displaces the state vector from the equilibrium state resulting from the prior collision step. The subsequent collision can therefore function as a \textit{denoising} operator mapping the perturbed state to the equilibrium. Let $\sqrt{\mathbf{f}} \in \mathbb{R}^q$ represent the post-streaming amplitude vector. Ideally, we seek a linear operator $\mathcal{D}^*$ that minimizes the residual distance $\| D \sqrt{\mathbf{f}} - \sqrt{\mathbf{f}^{\text{eq}}(\rho^{\text{str}}, \mathbf{u}^{\text{str}})} \|$, recovering the equilibrium state defined by the post-streaming macroscopic quantities $(\rho^{\text{str}}, \mathbf{u}^{\text{str}})$.

Since these target quantities are unknown a priori, we approximate $\mathcal{D}^*$ by a computable \textit{denoising operator} $\mathcal{D}$. We define $\mathcal{D}$ as the orthogonal projector onto the local linearization of the equilibrium manifold $\mathcal{M}$, an embedded manifold in $\mathbb{R}^q$ parameterized by $\mathbf{g}(\rho,\mathbf{u}) = \sqrt{\rho} \, \mathbf{h}(\mathbf{u})$. The local tangent space of the manifold is determined by the partial derivatives, $T_{({\rho}, {\mathbf{u}})}\mathcal{M} = \mathrm{span}\{\partial_{\rho} \mathbf{g}, \partial_{u_1} \mathbf{g}, \dots, \partial_{u_d} \mathbf{g}\}$. The scale invariance $T_{(\rho, \mathbf{u})}\mathcal{M} = T_{(c\rho, \mathbf{u})}\mathcal{M}$ for $c>0$ implies that the tangent space is fully determined by the velocity and independent of the density. Hence, we can construct a basis for this subspace solely using a reference velocity $\hat{\mathbf{u}}$:
\begin{align}
\partial_\rho \mathbf{g}(\rho, \hat{\mathbf{u}}) &\propto \mathbf{h}(\hat{\mathbf{u}}), \\
\partial_{u_k} \mathbf{g}(\rho, \hat{\mathbf{u}}) &\propto \partial_{k} \mathbf{h}(\hat{\mathbf{u}}).
\end{align}

The denoising operator is then defined as the orthogonal projector $\mathcal{D} = \bar{J} (\bar{J}^\top \bar{J})^{-1} \bar{J}^\top$ onto the column space of the scaled Jacobian $\bar{J} = [\mathbf{h}(\hat{\mathbf{u}}), \partial_{1} \mathbf{h}(\hat{\mathbf{u}}), \dots, \partial_{d} \mathbf{h}(\hat{\mathbf{u}})]$, which constitutes the local tangent space at $\hat{\mathbf{u}}$.

We now address the geometric validity of this projector. Specifically, we justify why the operator projects onto the linear subspace $T_{(\rho,\hat{\mathbf{u}})}\mathcal{M}$ rather than the affine tangent space $\mathbf{g}(\rho,\hat{\mathbf{u}}) + T_{(\rho,\hat{\mathbf{u}})}\mathcal{M}$, which constitutes the flat plane tangent to the smooth manifold. Geometrically, a local linearization of a manifold corresponds to an affine tangent space. However, the affine offset becomes redundant because the position vector $\mathbf{g}$ itself is collinear with the tangent vector along the density coordinate: $\mathbf{g}(\rho, \hat{\mathbf{u}}) \propto \mathbf{h}(\hat{\mathbf{u}}) \propto \partial_{\rho} \mathbf{g}(\rho, \hat{\mathbf{u}})$. Consequently, the affine tangent space is identical to the linear tangent space. In summary, $\mathcal{D}$ represents a projection onto the local linear expansion of the manifold, depending solely on the reference velocity while remaining invariant across all density scales. This geometric interpretation is illustrated in Fig.~\ref{fig:manifold_projection}. The vector $\mathcal{D}\sqrt{\mathbf{f}}$ is the unique point within the tangent subspace defined by $\hat{\mathbf{u}}$ that minimizes the distance to the input $\sqrt{\mathbf{f}}$, thereby filtering out noise components orthogonal to the linearized manifold.

\begin{figure}[!ht]
    \centering
    \includegraphics[width=1.\linewidth]{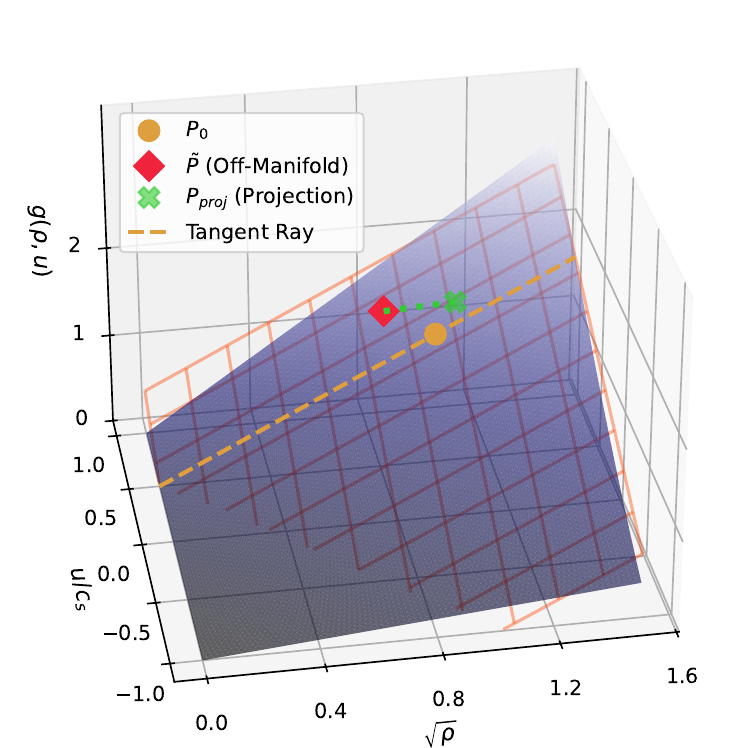}
    \caption{Projection onto the local tangent plane (orange grid) of a smooth surface (blue) in a neighborhood of the reference point $P_0$. Because of the conical structure in $\sqrt{\rho}$, the tangent plane touches the surface along a ray corresponding to a fixed value of $u$. Hence, the scaling factor $\sqrt{\rho}$ plays no role in constructing the tangent plane and its associated projection.}
    \label{fig:manifold_projection}
\end{figure}

In an LBM simulation, for each collision step, we construct a denoising operator from a selected reference velocity $\hat{\mathbf{u}}(t)$. It is also possible to specify a position-varying reference velocity $\hat{\mathbf{u}}(\mathbf{x},t)$. In that case, the cost of implementing it using elementary quantum gates depends on how the reference velocity varies across lattice sites. Since $\mathcal{D}$ is not a unitary, one has to use a block encoding of $\mathcal{D}$ to apply the denoising operator to the quantum state. We will discuss this problem later in the gate-level description of the collision step.

We can simplify the calculation with the factorization $\mathbf{h}(\mathbf{u}) = B \, \boldsymbol{\gamma}(\mathbf{u})$, where $B = [\mathbf{v}_0, \{\mathbf{v}_{k}\}_{k=1}^d, \{\mathbf{v}_{kl}\}_{1\leq k\leq l \leq d}] \in \mathbb{R}^{q \times d_{\mathcal{V}}}$ contains the basis vectors and $\boldsymbol{\gamma}(\mathbf{u}) \in \mathbb{R}^{d_{\mathcal{V}}}$ is their coefficients as given in Eq.~\eqref{eq:velocity_component}. The Jacobian matrix $\bar{J}$ can be written as $\bar{J} = B \,\Gamma(\hat{\mathbf{u}})$, where $\Gamma(\hat{\mathbf{u}}) = [\boldsymbol{\gamma}, \partial_{1} \boldsymbol{\gamma}, \dots, \partial_{d} \boldsymbol{\gamma}]_{\mathbf{u} = \hat{\mathbf{u}}} \in \mathbb{R}^{d_{\mathcal{V}} \times (d+1)}$. Since $B^\top B = I_{d_{\mathcal{V}}}$, the orthogonal projector for the specified reference velocity becomes
\begin{equation}
    \mathcal{D}(\hat{\mathbf{u}}) = B \, \Gamma \, (\Gamma^\top \Gamma)^{-1} \, \Gamma \, B^\top, \quad \Gamma = \Gamma(\hat{\mathbf{u}})
\label{eq:projection_hydro}
\end{equation}

\paragraph{Adaptation to the advection-diffusion problem}
The proposed framework readily extends to the LBM for the advection-diffusion equation (ADE) for a passive scalar field $C(\mathbf{x}, t)$ under an incompressible flow with advection velocity $\mathbf{u}_{\text{adv}}$,
\begin{equation}
\label{eq:advection_diffusion}
\partial_t C + \nabla \cdot (C \mathbf{u}_{\text{adv}}) = \kappa \nabla^2 C,
\end{equation}
where $\kappa$ is the diffusion constant. In the LBM formulation, $\kappa$ is related to the scalar relaxation time $\tau_g$ via $\kappa = c_s^2 (\tau_g - 1/2)$. Unlike the hydrodynamic case, only the zeroth moment $C(\mathbf{x}) = \sum_{i} f_i(\mathbf{x})$ is conserved during collision. The macroscopic velocity is not recovered from the distribution moments but is strictly imposed by the advection velocity field, $\mathbf{u} \equiv \mathbf{u}_{\text{adv}}$. Regarding the equilibrium distribution, while a first-order expansion is often utilized for ADE for its simplicity, we employ the full quadratic equilibrium expansion typically used for hydrodynamics. As noted in \cite{kruger_lattice_2017}, truncating the equilibrium at the first order introduces a velocity-dependent effective diffusion term. Retaining the quadratic terms ensures that the physical diffusion remains independent of $\mathbf{u}_{\text{adv}}$.

The denoising operator requires modifications for simulating advection-diffusion problems. Since we imposed the macroscopic velocity through $\mathbf{u}_{\text{adv}}$, the scaled Jacobian $\bar{J}$ and the denoising operator $\mathcal{D}$ no longer depend on an arbitrary reference. Because the advection velocity is assumed to remain invariant during collision, partial derivatives with respect to velocity are irrelevant. The scaled Jacobian and the denoising operator are parameterized solely by the advection velocity
\begin{equation}
\bar{J}_{\mathrm{ADE}} = \mathbf{h}(\mathbf{u}_{\text{adv}}), \quad \mathcal{D}_{\mathrm{ADE}} = \frac{\mathbf{h}(\mathbf{u}_{\text{adv}}) \, \mathbf{h}(\mathbf{u}_{\text{adv}})^\top}{\|\mathbf{h}(\mathbf{u}_{\text{adv}})\|^2}
\label{eq:projection_ADE}
\end{equation}

\paragraph{Equivariance}
This part demonstrates that the denoising operator satisfies inherent symmetry constraints of a collision process in LBM. Let $\Omega$ be a general collision operator such that $\Omega[\mathbf{f}] \in \mathbb{R}^q$ denotes the post-collision population vector. First, the collision operator satisfies a first-order homogeneity property called \emph{scale equivariance}: $\Omega[\lambda \mathbf{f}] = \lambda\,\Omega[\mathbf{f}]$ for all $\lambda \ge 0$. Second, it must be compatible with rotations and reflections of the lattice, a property called \emph{equivariance to lattice symmetries} \cite{corbetta_toward_2023}. Most lattice configurations $\{(\mathbf{c}_i, w_i)\}_{i=1}^q$ possess a symmetry group that preserves the discrete velocity set. For example, the D2Q9 model admits the dihedral symmetry group of the square, $D_8 = \langle r,s \mid r^4 = s^2 = e,\; srs = r^{-1}\rangle$, generated by a $90^\circ$-rotation $r$ and a reflection $s$ over an axis. Writing $g \cdot \mathbf{f}$ for the permutation of the components of $\mathbf{f} = (f_1,\dots,f_q)$ induced by a symmetry $g \in D_8$, an equivariant relation is given by
\begin{equation}
\label{eq:collision_symm_equiv}
\Omega[g \cdot \mathbf{f}] = g \cdot \Omega[\mathbf{f}].
\end{equation}
This identity ensures that pre- and post-collision populations transform consistently under symmetry actions.

The equivariant relations extend to the denoising operator. The scale equivariance follows directly from the operator's linearity, while Property~\ref{prop:equivariance} establishes equivariance under lattice symmetries. Thus, the denoising operator satisfies all symmetry constraints of a collision operator.

\begin{property}[Symmetry equivariance]
\label{prop:equivariance}
Let $G$ be the symmetry group of the lattice. Suppose a symmetry $g \in G$ acts on population vectors through the permutation matrix $P_g \in \mathbb{R}^{q\times q}$ and on physical velocities through the orthogonal map $R_g \in O(d)$. The denoising operator satisfies the equivariant relation
\begin{equation}
\mathcal{D}(R_g \hat{\mathbf{u}}) P_g = P_g  \mathcal{D}(\hat{\mathbf{u}})
\end{equation}
\end{property}
\begin{proof}
Appendix~\ref{appdx:equivariance}.
\end{proof}

\paragraph{Error bounds}
The following results establish error bounds for the denoising operator. Theorem \ref{thm:collision_error} estimates the error incurred when using $\mathcal{D}$ in the collision step. This collision error scales with $\|\Delta \mathbf{u}\|^2$, the squared distance between the reference velocity and the actual velocity, and also depends on the rate of change of the velocity across the domain.

\begin{theorem}[Collision error]
\label{thm:collision_error}
Let $\tilde{\mathbf{g}} \in \mathbb{R}^q$ be the post-streaming amplitude vector with macroscopic quantities $(\rho^{\text{str}}, \mathbf{u}^{\text{str}})$. Denote $\mathcal{E} = \| \mathcal{D} \, \tilde{\mathbf{g}} - \mathbf{g}(\rho^{\text{str}}, \mathbf{u}^{\text{str}}) \|$ as the collision error of $\mathcal{D}$. Then,
\begin{equation}
    \mathcal{E} \leq \frac{\sqrt{\rho^{\text{str}}}}{2} \left( \|\mathbf{H}\|_2 \|\Delta \mathbf{u}\|^2 + \sqrt{2} \|S(\mathbf{u}^{\text{str}})\|_{\mathrm{F}} \right),
\end{equation}
where $\Delta \mathbf{u} = \mathbf{u}^{\text{str}} - \hat{\mathbf{u}}$ is the velocity misalignment, $\mathbf{H} = \nabla^2\mathbf{h}$ the constant Hessian tensor of the quadratic function $\mathbf{h}(\mathbf{u})$, and $S(\mathbf{u}^{\text{str}}) = \frac{1}{2} \left(\nabla \mathbf{u}^{\text{str}} + (\nabla \mathbf{u}^{\text{str}})^\top \right)$ the \textit{strain-rate tensor}.
\end{theorem}
\begin{proof}
    See Appendix \ref{appdx:error_analysis} for the proof.
\end{proof}

Theorem \ref{thm:denoising-error}, on the other hand, estimates the Euclidean distance between the projection and the equilibrium manifold. Briefly, this out-of-equilibrium error scales quadratically with the distance between the initial state and the reference state on the manifold.

\begin{theorem}[Denoising error]
\label{thm:denoising-error}
Let $\tilde{\mathbf{g}} \in \mathbb{R}^q$ be an arbitrary vector. Define $\mathcal{E}^\perp = \min_{\rho, \mathbf{u}} \| \mathcal{D} \, \tilde{\mathbf{g}} - \mathbf{g}(\rho, \mathbf{u})\|$ as the distance from the projection to the equilibrium manifold. Given a reference density  $\hat\rho$, this distance is bounded, under conditions for Newton-Kantorovich convergence, by:
\begin{equation}
    \mathcal{E}^\perp \leq 2 \beta^2 \Delta^2 \|\mathbf{H}\|_2 \sqrt{\hat \rho + 2\beta \Delta},
\end{equation}
where $\Delta =  \|\sqrt{\hat \rho} \mathbf{h}(\hat{\mathbf{u}}) - \tilde{\mathbf{g}}\|$ and $\beta =  1 / \sigma_{\text{min}} (\nabla \mathbf{g}(\hat\rho, \hat{\mathbf{u}}))$ is the inverse of the smallest singular value of the full-column-rank Jacobian. Furthermore, choosing $\sqrt{\hat{\rho}} = \mathbf{h}(\hat{\mathbf{u}})^\top \tilde{\mathbf{g}} / \|\mathbf{h}(\hat{\mathbf{u}})\|^2$ minimizes $\Delta$, yielding
\begin{equation}
    \Delta = \|\tilde{\mathbf{g}}\| \left|\sin\theta\right|, \quad \cos\theta = \frac{\mathbf{h}(\hat{\mathbf{u}})^\top \tilde{\mathbf{g}}}{\|\mathbf{h}(\hat{\mathbf{u}})\| \|\tilde{\mathbf{g}}\|}
\end{equation}
This gives a sharper bound for $\mathcal{E}^\perp$ in terms of the misalignment angle $\theta$.
\end{theorem}
\begin{proof}
    See Appendix \ref{appdx:error_analysis} for the proof, including the necessary conditions.
\end{proof}

\paragraph{Quantum gate decomposition}
This part explains how to implement the collision step with the denoising operator as a quantum circuit. The key challenge is the fact that the denoising operator is a non-unitary that acts on elements in $\mathbb{R}^q$. We will explain how to embed it in a unitary operator acting on $\mathcal{H}_V \otimes \mathbb{C}^2$, the composite Hilbert space of $V$ and an ancilla. A detailed gate-level implementation of the collision step will also be provided.

First, we encode the denoising operator in a unitary matrix that acts on $\mathbb{R}^q \otimes \mathbb{R}^2 \cong \mathbb{R}^{2q}$ using the block-encoding method introduced by \textcite{schlimgen_quantum_2022}. The method applies for any square matrix $\mathcal{A} \in \mathbb{R}^{q \times q}$ such that $\|\mathcal{A}\| \leq 1$ (otherwise, normalize $\mathcal{A}$ by $\alpha \geq \|\mathcal{A}\|$). Let it have the singular value decomposition $\mathcal{A} = U \Sigma V^\dagger$, where $U,V^\dagger \in \mathbb{R}^{q \times q}$ are unitary matrices and $\Sigma = \mathrm{diag}(\sigma_{1}, \dots, \sigma_q)$ contains the singular values of $\mathcal{A}$. Define two diagonal matrices $\Sigma_{+}$ and $\Sigma_{-}$ by
\begin{equation}
    \Sigma_{\pm} = \mathrm{diag}(\sigma_{1} \pm i\sqrt{1 - \sigma_1^2}, \dots, \sigma_{q} \pm i\sqrt{1 - \sigma_q^2})
\end{equation}
As both $\Sigma_{+}$ and $\Sigma_{-}$ are unitary, the diagonal concatenation $U_{\Sigma} = \mathrm{diag}(\Sigma_{+}, \Sigma_{-})$ is also unitary. Define the block-encoding of $\mathcal{A}$ as 
\begin{equation}
    U_{\mathcal{A}} = (U \otimes H) U_{\Sigma} (V^\dagger \otimes H),
\end{equation}
where $H$ is the Hadamard operator. One can verify that the top-left block of the unitary is $(I \otimes \bra{0}) U_{\mathcal{A}} (I \otimes \ket{0}) = \frac{1}{2} U (\Sigma_+ + \Sigma_-) V^\dagger =\mathcal{A}$. Specific to our case, when $\mathcal{A} = \mathcal{D}$ is a projector, the singular values are either $0$ or $1$. As a result, the diagonal entries of $U_{\Sigma}$ can only be $1$, $i$, or $-i$.

\begin{figure}[!ht]
    \centering
    \begin{subfigure}[b]{1.0\linewidth}
        \centering
        \includegraphics[width=\linewidth]{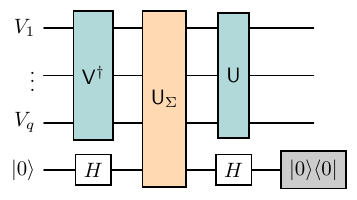}
        \caption{}
        \label{fig:collision_qc}
    \end{subfigure}
    
    \vspace{0.5cm} 
    
    \begin{subfigure}[b]{1.0\linewidth}
        \centering
        \includegraphics[width=\linewidth]{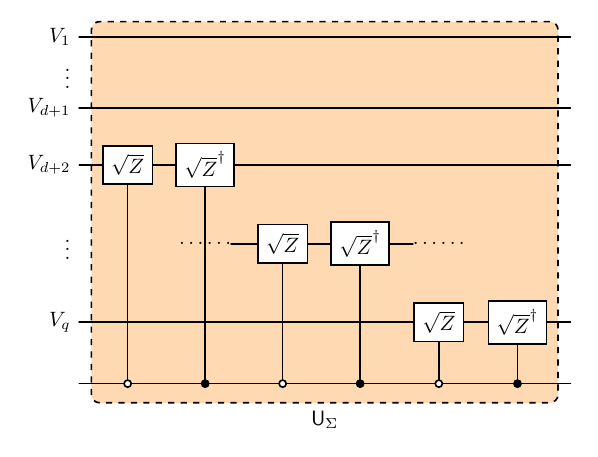}
        \caption{}
        \label{fig:collision_phase}
    \end{subfigure}
    \caption{\textbf{(a)} Quantum implementation of the collision step, including the collision unitary operator $\mathsf{U}_{\mathcal{D}}$ followed by a post-selection of the ancilla. Every collision step must be supplied with an ancilla in the state $\ket{0}$, which is detached upon successful post-selection in $\ket{0}$. In numerical simulations, the post-collision state in $\mathcal{H}_V$ must be renormalized to ensure density conservation. \textbf{(b)} Diagonal operator $\mathsf{U}_{\Sigma}$ for flow problems as a sequence of controlled phase gates, assuming non-trivial phases on the last $q-d-1$ qubits.}
    \label{fig:collision_step}
\end{figure}

We construct a unitary quantum operator on $\mathcal{H}_V \otimes \mathbb{C}^2$ as an extension of $U_{\mathcal{D}}$ by identifying the standard basis of $\mathbb{R}^q$ with the one-hot basis $\{\ket{\mathbf{e}_1}, \dots, \ket{\mathbf{e}_q}\}$ while keeping other basis states of $\mathcal{H}_V$ unaffected. This corresponds to the one-hot extension of unitary matrices $U$ and $V^\dagger$, denoted by upright notations $\mathsf{U}$ and $\mathsf{V}^\dagger$, respectively. Starting from the identity of $\mathcal{H}_V$, the extension is the replacement of entries specified by the rows and columns of the one-hot states $\{\ket{\mathbf{e}_1}, \dots, \ket{\mathbf{e}_q}\}$ with $U$ or $V^\dagger$. Similarly, let  $\mathsf{U}_{\Sigma}$ be the one-hot extension of $U_{\Sigma}$ on $\mathcal{H}_V \otimes \mathbb{C}^2$. It is also a diagonal operator and can be implemented by controlled phase gates, which only put the non-trivial phases $\pm i$ to some basis states. The number of non-trivial phases is $2(q - \mathrm{rank}(\mathcal{D}))$, which equals $2(q-d-1)$ for flow problems and $2(q-1)$ for advection-diffusion problems. The unitary operator for collision is implemented by
\begin{equation}
    \mathsf{U}_{\mathcal{D}} = (\mathsf{U} \otimes H) \mathsf{U}_{\Sigma} (\mathsf{V}^\dagger \otimes H)
\end{equation}
The collision step is achieved with an application of $\mathsf{U}_{\mathcal{D}}$ and an ancilla post-selection, as displayed in Fig.~\ref{fig:collision_step}. The figure also presents an implementation of the diagonal operator $\mathsf{U}_{\Sigma}$ using controlled phase gates. Meanwhile, $\mathsf{U}$ and $\mathsf{V}^\dagger$ can be implemented using standard gate-decomposition methods \cite{mottonen_quantum_2004, vartiainen_efficient_2004, rakyta_approaching_2022}. However, since $\mathsf{U}$ and $\mathsf{V}^\dagger$ only act non-identically on the one-hot basis states, it is more efficient and numerically accurate to implement them with \textit{Givens rotations}. Real-valued Givens rotations are particle-preserving unitaries acting on the two-dimensional subspace spanned by $\{\ket{01},\ket{10}\}$ and take the form \cite{arrazola_universal_2022}
\begin{equation}
    G(\theta) = 
    \begin{bmatrix}
    1 & 0 & 0 & 0 \\
    0 & \cos\theta & - \sin\theta & 0 \\
    0 & \sin\theta & \cos\theta & 0 \\
    0 & 0 & 0 & 1
    \end{bmatrix}
\end{equation}
\citeauthor{clements_optimal_2016}~\cite{clements_optimal_2016, cilluffo_commentary_2024} proposed an algorithm that decomposes a real-valued $N \times N$ unitary as $D \prod_{(m,n) \in \mathcal{N}} G_{mn}$, where $D$ is a diagonal phase matrix and the subscripts $m,n$ denote basis element indices for the 2D subspace following an ordered sequence $\mathcal{N}$. In fact, the algorithm produces a sequence of exchanges between consecutive basis elements, i.e. $n = m+1$, and $D$ can be fixed to the identity, absorbing all its phases $\pm 1$ to the Givens rotations. The decomposition achieves a minimal number of rotations, $N(N-1)/2$, arranged in a minimal depth by laying rotation operators in a brickwork structure. In our setting, we can apply this decomposition to the effective unitaries $U$ and $V^\dagger$ in $\mathbb{R}^q$ instead of their one-hot extensions. This results in $q(q-1)/2$ effective Givens rotations for each unitary acting on $\mathbb{R}^q$. The Givens rotations are implemented on quantum circuits by their one-hot extensions on the full Hilbert space $\mathcal{H}_{V}$. Fig.~\ref{fig:collision_subblock} shows the general layout of the Clements decomposition into Givens rotations.

\begin{figure}[!ht]
    \centering
    \includegraphics[width=\linewidth]{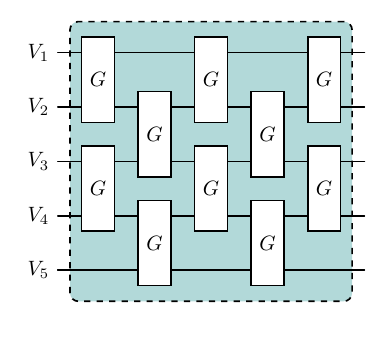}
    \caption{The Clements decomposition of $\mathsf{U}$ and $\mathsf{V}^\dagger$ yields a brickwork layout of Givens rotations, with each rotation having a distinct angle. This diagram is shown for $q = 5$.}
    \label{fig:collision_subblock}
\end{figure}

To sum up, the implementation of $\mathsf{U}_{\mathcal{D}}$ employs $2(q - \mathrm{rank}(\mathcal{D}))$ controlled phase gates and $q(q-1)$ Givens rotations. The ancilla must be post-selected in the state $\ket{0}$. This is done by measuring the ancilla and retaining the normalized resulting state only when the outcome is $0$. If the outcome is $1$, the experiment is discarded and the LBM simulation must be restarted. 

\subsection{Streaming}
The streaming step in Eq.~\eqref{eq:lbm_streaming} can be implemented by a conditional shift operator acting on the spatial register, $\mathcal{S}: \ket{\mathbf{x}} \ket{\mathbf{e}_i} \mapsto \ket{\mathbf{x} + \mathbf{c}_i} \ket{\mathbf{e}_i}$. There are numerous quantum arithmetic techniques that have been adopted for efficient implementations of this streaming operator \cite{todorova_quantum_2020, schalkers_efficient_2024}. In light of that development, we only describe the advantage of one-hot encoding in the construction of the streaming operator and provide a decomposition using the quantum addition $\ket{x} \mapsto \ket{x+1}$ as the unit resource, which is open to further gate-level optimizations.

A distinct benefit of using one-hot encoding for the local state (Eq.~\eqref{eq:local_state}) is the simplified quantum control logic. Since the velocity state $\ket{\mathbf{e}_i}$ is uniquely identified by the activation of the $i$-th qubit in the velocity register (denoted as $V_i$), the conditional shift does not require multiple control qubits. Instead, the arithmetic addition, denoted as $\mathsf{ADD}(\mathbf{c}_i)$, is controlled solely by the state of the single qubit $V_i$. Therefore, the streaming operator decomposes as
\begin{equation}
\mathcal{S} = \prod_{i=1}^q \left( \mathsf{ADD}(\mathbf{c}_i) \otimes \ket{1}\bra{1}_{V_i} + \mathsf{I}_D \otimes \ket{0}\bra{0}_{V_i} \right)
\label{eq:streaming-op},
\end{equation}
which ensures that when the velocity register is in state $\ket{\mathbf{e}_i}$, only the $i$-th term in the product applies a shift $\mathbf{c}_i$ to the domain register $D$. This formulation represents the streaming step in the general LBM simulation workflow (Fig.~\ref{fig:LBM_workflow}). We can further optimize the circuit depth using the structure of the discrete velocities. In common velocity sets, each velocity vector $\mathbf{c}_i$ can be decomposed into a sum of unit steps in Cartesian axes ($\pm \hat{x}, \pm \hat{y}$ and $\pm \hat{z}$). By factoring $\mathsf{ADD}(\mathbf{c}_i)$ into unit shifts, we can rearrange the streaming step to apply these shifts sequentially. Instead of controlling a shift $\mathbf{c}_i$ by a single qubit $V_i$, we group the velocity states based on their spatial components. Let $\hat\Gamma = \{ \pm \hat{x}, \pm \hat{y}, \dots \}$ be the set of $2d$ unit directions in a $d$-dimensional lattice. The streaming operator can be rewritten as a product of controlled unit shifts,
\begin{equation}\mathcal{S} = \prod_{\hat{\mathbf{c}} \in \hat{\Gamma}} \left( \mathsf{ADD}(\hat {\mathbf{c}}) \otimes \Pi_{\hat {\mathbf{c}}} + \mathsf{I}_D \otimes (\mathbb{I}_V - \Pi_{\hat {\mathbf{c}}}) \right), \label{eq:streaming-opt}
\end{equation}
where $\Pi_{\hat {\mathbf{c}}}$ is a projector onto the subspace of velocity states that include the component $\hat {\mathbf{c}}$. In the one-hot representation, this projector corresponds to the logical $\mathsf{OR}$ of the specific velocity qubits that share the direction $\hat {\mathbf{c}}$, i.e., $\Pi_{\hat {\mathbf{c}}} = \sum_{i : \mathbf{c}_i \cdot \hat {\mathbf{c}} > 0} \ket{1}\bra{1}_{V_i}$. Fig.~\ref{fig:streaming_qc} shows an implementation of this reformulated streaming operator using an ancilla to store the logical $\mathsf{OR}$ and serve as the control qubit for the unit shifts. We note that for a one-hot bitstring $b_1b_2\dots b_q \in \{0,1\}^q$, the control variable is $b_{\mathrm{anc}} = \bigvee_{i : \mathbf{c}_i \cdot \hat {\mathbf{c}} > 0}b_i = \bigoplus_{i : \mathbf{c}_i \cdot \hat {\mathbf{c}} > 0}b_i$, which can be computed by applying $\mathrm{CNOT}$ gates on every qubit $V_i$ with $\mathbf{c}_i \cdot \hat {\mathbf{c}} > 0$ and the ancilla.

Finally, we demonstrate that the square-root amplitude encoding $\sqrt{f_i}$ is necessary to simulate the correct macroscopic behavior during propagation. The action of $\mathcal{S}$ on the state yields:
\begin{align}\mathcal{S}\ket{\psi} &= \sum_{\mathbf{x}}\sum_i \sqrt{f_i(\mathbf{x} - \mathbf{c}_i)} \ket{\mathbf{x}} \ket{\mathbf{e}_i} \nonumber \\
&= \sum_{\mathbf{x}} \sqrt{\rho^{\text{str}}(\mathbf{x})} \ket{\mathbf{x}} \ket{\psi^{\text{str}}_{\mathbf{x}}}
\end{align}
The resulting macroscopic density is $\rho^{\text{str}}(\mathbf{x}) = \sum_i | \langle{\mathbf{x}, \mathbf{e}_i | \mathcal{S} | \psi \rangle} |^2 = \sum_i f_i(\mathbf{x}-\mathbf{c}_i)$, consistent with the transport dynamics of the standard LBM. However, using the first-order local encoding $|\tilde\psi_{\mathbf{x}}\rangle = [f_1, \dots, f_q]^\top / \big(\sum_i f_i^2\big)^{1/2}$ yields $\rho^{\text{str}}(\mathbf{x}) = \sum_i \rho(\mathbf{x}-\mathbf{c}_i) \frac{f^2_i(\mathbf{x}-\mathbf{c}_i)}{\|\mathbf{f}(\mathbf{x}-\mathbf{c}_i)\|^2}$ which does not reproduce the correct linear transport dynamics.

\begin{figure}[!ht]
    \centering
    \includegraphics[width=\linewidth]{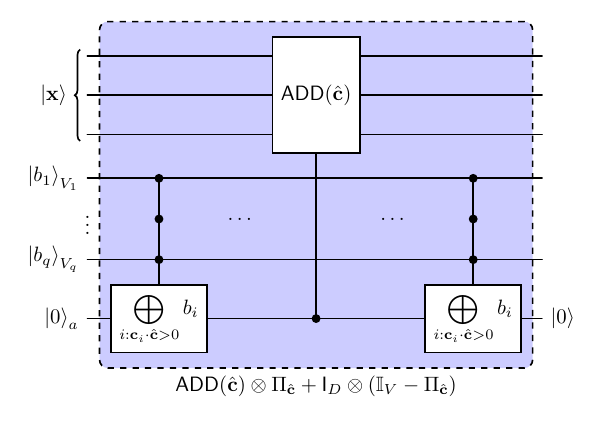}
    \caption{General quantum circuit for streaming along a Cartesian unit direction $\hat{\mathbf{c}}$. The circuit computes and uncomputes the binary value for control with an ancilla. The complete streaming process can be performed by concatenating $2d$ similar circuits.}
    \label{fig:streaming_qc}
\end{figure}

\subsection{Boundary conditions}
Quantum implementations of the LBM often accommodate two types of boundary conditions: periodic boundary conditions and bounce-back. We will briefly explain the corresponding physical behavior and implementations.

Periodic boundary conditions treat opposite edges as being glued together; a quantity leaving the domain at one side will appear at the opposite. This enforces a global density conservation and aligns naturally with reversible quantum evolutions. The realization of periodic boundary is subsumed in the streaming step thanks to quantum modular arithmetic operations.

Bounce-back is when fluid distributions are reflected at the reverse direction upon encountering walls or objects during streaming. This is known to approximate the no-slip boundary condition $v_{\text{wall}} = 0$. A bounce-back scheme is called \textit{full-way} when the boundary surface is positioned at the most exterior solid node adjacent to a fluid node and \textit{half-way} when the surface is exactly halfway between a fluid node and a solid node. Fig.~\ref{fig:bounce-back_scheme} contains schematic diagrams demonstrating the effects of both bounce-back schemes. For the full-way bounce-back (FWBB), a fluid distribution simply reverses its velocity direction upon encountering a solid node and skips the next collision step, since collision does not happen at a solid node. For the half-way bounce-back (HWBB), a fluid distribution encountering a solid node is reversed in velocity direction and returns to the fluid node prior to streaming. The HWBB achieves second-order spatial accuracy for flat-wall boundaries, while the FWBB only has first-order spatial accuracy \cite{burel_improved_2018, ziegler_boundary_1993}.

\begin{figure}[!ht]
    \centering
    \includegraphics[width=\linewidth]{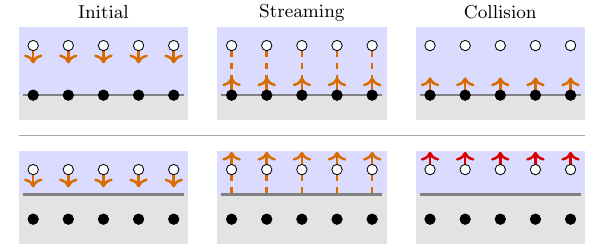}
    \caption{Schematic illustration of the \textbf{(top)} full-way and \textbf{(bottom)} half-way bounce-back schemes. White and black nodes denote fluid and solid nodes, respectively, and the flat plane marks the fluid-solid interface. The same dynamics apply in every streaming direction. In the bottom-right panel, the color change indicates the post-collision state.}
    \label{fig:bounce-back_scheme}
\end{figure}

We propose circuit implementations for the bounce-back schemes. Unlike the wall-free domain setting where a single QLBM step includes collision and streaming as two separate processes (Fig.~\ref{fig:LBM_workflow}), implementing the bounce-back schemes requires unifying collision, streaming, and other specific operators in each iteration. Let $\mathcal{C}$ be the collision operator that acts on $\mathcal{H}_V$ (the involved ancilla is ignored as it is initialized and post-selected in $\ket{0}$ upon successful collision) and $\mathcal{S}$ the streaming operator on $\mathcal{H}_D \otimes \mathcal{H}_V$. We define the direction reversal operator $\mathcal{R}: \ket{\mathbf{e}_i} \mapsto \ket{\mathbf{e}_{\bar i}}$, where $\bar{i}$ is the index of the reverse velocity, $\mathbf{c}_{\bar i} = - \mathbf{c}_i$. Since $\bar{\bar i} = i$, $\mathcal{R}$ can be implemented by $\mathsf{SWAP}$ gates between every qubit pair $(i,\bar i)$. In common lattices where velocities come in opposite pairs except for the zero velocity, $(q-1)/2$ $\mathsf{SWAP}$ gates are required. We also assume access to a query oracle that checks if a node is a solid node $\mathcal{Q}: \ket{\mathbf{x}}\ket{b} \mapsto \ket{\mathbf{x}} \ket{b \, \oplus \text{is\_solid}(\mathbf{x})}$. Denote $B$ as the ancilla whose binary state $\ket{b}$ indicates solid nodes. Bounce-back cannot be implemented as a standalone operator but rather integrated into a QLBM step that includes collision, streaming, and direction reversal. The FWBB and HWBB schemes are realized in an LBM step by
\begin{equation}
\mathcal{U}_{\text{FWBB}} = \mathcal{Q}^\dagger (|0\rangle\langle 0|_B \otimes \mathcal{C} + |1\rangle\langle 1|_B \otimes \mathcal{R}) \mathcal{Q} \, \mathcal{S}
\end{equation}
and 
\begin{equation}
    \mathcal{U}_{\text{HWBB}} = \mathcal{C} \mathcal{S} \mathcal{Q}^\dagger (|0\rangle\langle 0|_B \otimes \mathcal{S}^\dagger + |1\rangle\langle 1|_B \otimes \mathcal{R}) \mathcal{Q} \mathcal{S}.
\end{equation}
Fig.~\ref{fig:bounce-back} shows the quantum circuits for the two bounce-back schemes. It is straightforward to verify their action:
\begin{multline}
    \mathcal{U}_{\text{FWBB}} \ket{\mathbf{x}} \ket{\mathbf{e}_i} \ket{0} \qquad \qquad \qquad \qquad \qquad \qquad \\
= \left[ w\ket{\mathbf{x} + \mathbf{c}_i} \ket{\mathbf{e}_{\bar{i}}} + (1-w)\ket{\mathbf{x} + \mathbf{c}_i} \mathcal{C} \ket{\mathbf{e}_i}\right] \ket{0}
\end{multline}
and 
\begin{multline}
    \mathcal{U}_{\text{HWBB}} \ket{\mathbf{x}} \ket{\mathbf{e}_i} \ket{0} \qquad \qquad \qquad \qquad \qquad \qquad \\
= \left[ w \ket{\mathbf{x}}\mathcal{C}\ket{\mathbf{e}_{\bar{i}}} + (1-w) \ket{\mathbf{x} + \mathbf{c}_i}\mathcal{C}\ket{\mathbf{e}_i} \right] \ket{0},
\end{multline}
where $w = \text{is\_solid}(\mathbf{x} + \mathbf{c}_i) \in \{0,1\}$.

\begin{figure}[!ht]
    \centering
    \begin{subfigure}[b]{0.7\linewidth}
        \centering
        \includegraphics[width=\linewidth]{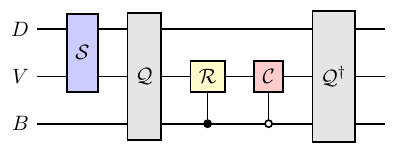}
        \caption{}
        \label{fig:FWBB}
    \end{subfigure}
    
    \vspace{0.5cm} 
    
    \begin{subfigure}[b]{1.0\linewidth}
        \centering
        \includegraphics[width=\linewidth]{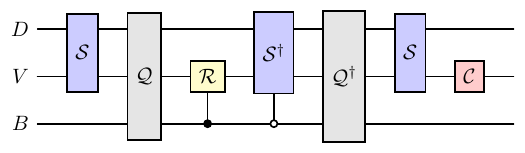}
        \caption{}
        \label{fig:HWBB}
    \end{subfigure}
    
    \caption{Quantum circuit for a full LBM step including streaming, collision, and \textbf{(a)} full-way bounce-back \textbf{(b)} half-way bounce-back.}
    \label{fig:bounce-back}
\end{figure}

\subsection{Extracting physical quantities}
In many applications, one may wish to extract physical quantities from the simulation. An example is the total force exerted by the fluid on a solid region $\mathcal{B}$. In the lattice Boltzmann framework, this quantity is typically evaluated using the momentum exchange method \cite{schalkers_momentum_2024}. Let $\mathcal{L}(\mathcal{B})$ denote the set of all fluid-solid links crossing the boundary of $\mathcal{B}$; each link is specified by a pair $(\mathbf{x},i)$, where $\mathbf{x}$ is a fluid node adjacent to $\mathcal{B}$ and $\mathbf{c}_i$ is a discrete velocity pointing from $\mathbf{x}$ into $\mathcal{B}$. For such a link, the incoming population $f_i(\mathbf{x},t)$ travels along $\mathbf{c}_i$ toward the solid and is reflected to $\mathbf{c}_{\bar i}=-\mathbf{c}_i$ under bounce-back, resulting in the fluid momentum change $\Delta \mathbf{p}(\mathbf{x},i) =-2 f_i(\mathbf{x},t)\,\mathbf{c}_i$. By Newton’s third law, the momentum transferred to the solid is $2 f_i(\mathbf{x},t)\mathbf{c}_i$. Summing over all links intersecting the solid region and dividing by the time step yields the hydrodynamic force on $\mathcal{B}$,
\begin{equation}
\mathbf{F}(\mathcal{B}) = \frac{2}{\Delta t}
\sum_{(\mathbf{x},i)\in \mathcal{L}(\mathcal{B})}
f_i(\mathbf{x},t)\,\mathbf{c}_i .
\end{equation}

For each Cartesian component $k\in\{1,\dots,d\}$, define the force observable associated with the solid region $\mathcal{B}$,
\begin{equation}
\hat{F}^{(\mathcal{B})}_k
=
\frac{2}{\Delta t}
\sum_{(\mathbf{x},i)\in \mathcal{L}(\mathcal{B})}
c_{ik}\;
\ket{\mathbf{x}}\bra{\mathbf{x}}_{D}
\otimes
\ket{\mathbf{e}_i}\bra{\mathbf{e}_i}_{V}.
\end{equation}
Taking the expectation value with respect to the global state $\ket{\psi}$ gives a force component.
\begin{equation}
\langle \hat{F}^{(\mathcal{B})}_k\rangle_\psi = \frac{2}{\Delta t}
\sum_{(\mathbf{x},i)\in \mathcal{L}(\mathcal{B})} c_{ik}\, f_i(\mathbf{x},t) = F_k(\mathcal{B}),
\end{equation}
Operationally, measuring $\hat{F}^{(\mathcal{B})}_k$ corresponds to measuring $(\mathbf{x},i)$ in the position-velocity space, assigning the value $\frac{2}{\Delta t} c_{ik}$ when $(\mathbf{x},i)\in\mathcal{L}(\mathcal{B})$ and $0$ otherwise, and averaging over repeated samples. The total force vector $\mathbf{F}(\mathcal{B})$ is thus estimated through the expectation of the collective observable $\hat{\mathbf{F}}^{(\mathcal{B})}=(\hat{F}^{(\mathcal{B})}_1,\dots,\hat{F}^{(\mathcal{B})}_d)$.

\subsection{Approaches for fully coherent QLBM simulations}
In the implementation protocol of our QLBM algorithm, the collision step is realized by the projector $\mathcal{D}$. With an abuse of notation, we also use $\mathcal{D}$ to represent the extension of the projector to the subspace of one-hot basis states. Implementing such a non-unitary on quantum circuits involves block-encoding in an enlarged system followed by post-selection on a specific state of ancillas. However, this process is non-deterministic as there is a nonzero probability that ancilla measurements yield undesired states. Even though each collision step might have a large success probability, the cumulative success probability decreases exponentially over multiple time steps. This section briefly examines potential approaches to address this problem.

One approach is to raise the success probability to near unity using oblivious amplitude amplification (OAA). However, performing OAA effectively for non-unitary operators is still an open question \cite{zecchi_improved_2025}. In fact, we can show that amplitude amplification does not increase the success probability of applying the projector $\mathcal{D}$.

\begin{lemma}
Let $U$ be a unitary block-encoding such that $U(\ket{0} \otimes \ket{\varphi}) = \ket{0} \otimes K\ket{\varphi} + \ket{1} \otimes \ket{\text{junk}_{\varphi}}$. If $K$ is an orthogonal projector, then no oblivious amplitude amplification protocol $\tilde{U}$ can increase the success probability while preserving the action of $K$ up to a scalar.
\end{lemma}
\begin{proof}
    We first note that $K$ is the Kraus operator\footnote{A quantum channel $\Phi$ between Hilbert spaces $\mathcal{H}$ and $\mathcal{G}$ can be expressed via Kraus operators $\{B_i\}$ as $\Phi(\rho) = \sum_i B_i \rho B_i^\dagger$ such that $\sum_i B_i^\dagger B_i = I$.} defined by $K = (\bra{0} \otimes I)U(\ket{0} \otimes I)$. Physically, it represents the non-unitary transformation the system undergoes when the ancilla is successfully post-selected in the $\ket{0}$ state.
    
    An OAA protocol, given by a unitary $\tilde{U}$, defines a new Kraus operator $K'$ via $\tilde{U}(\ket{0} \otimes \ket{\varphi}) = \ket{0} \otimes K'\ket{\varphi} + \ket{1} \otimes \ket{\text{junk}'_{\varphi}}$. By the linearity of $\tilde{U}$ and the requirement that $K'\ket{\varphi} = \lambda(\varphi) K\ket{\varphi}$ for all $\ket{\varphi}$, the scalar must be a constant $\lambda$.

    Since $K'$ is a Kraus operator derived from a unitary and post-selection, its singular values $t_j$ must satisfy $|t_j| \leq 1$. Let $s_j$ be the singular values of $K$. The relation $K' = \lambda K$ implies $t_j = |\lambda| s_j$. Since $K$ is an orthogonal projector, its singular values are $s_j \in \{0, 1\}$. Choose $j$ where $s_j = 1$, then it follows that $|\lambda| \leq 1$. The success probability of the OAA process is $P_{\text{succ}}' = \|K'\ket{\varphi}\|^2 = |\lambda|^2 \|K\ket{\varphi}\|^2 \leq \|K\ket{\varphi}\|^2$. Thus, the success probability cannot be increased.
\end{proof}

Another approach is to implement the projection as a linear filter in terms of the Hamiltonian $H = I - \mathcal{D}$,
\begin{equation}
    \ket{\psi} \mapsto \frac{(I - H)\ket{\psi}}{\|(I - H)\ket{\psi}\|} = \frac{\mathcal{D} \ket{\psi}}{\|\mathcal{D} \ket{\psi}\|}.
\end{equation}
In the framework of \textit{double-bracket quantum algorithms} \cite{gluza_double-bracket_2026, suzuki_double-bracket_2025}, this mapping can be realized by the unitary operator $U_{\psi} = e^{s[\ket{\psi}\bra{\psi}, H]}$ with
\begin{equation}
    s_{\psi} = \frac{-1}{\sqrt{V_\psi}} \cos^{-1}\left( \frac{E_\psi - 1}{\sqrt{V_\psi + (E_\psi - 1)^2}}\right),
\end{equation}
where $E_\psi = \bra{\psi} H \ket{\psi}$ and $V_\psi = \bra{\psi} H^2 \ket{\psi} - E_\psi^2$. Since $H = I - \mathcal{D}$ is also a projector, it is straightforward to verify $s_{\psi} = - \cos^{-1}(-\sqrt{p}) / \sqrt{p(1-p)}$ with $p = \bra{\psi} \mathcal{D} \ket{\psi}$. The unitary $U_{\psi}$ can be approximated by the group commutator formula \cite{gluza_double-bracket_2024}
\begin{nalign}
     U_\psi =& \left( e^{i s_\psi^{(N)} \ket{\psi} \bra{\psi}} \,e^{i s_\psi^{(N)} H} e^{-i s_\psi^{(N)} \ket{\psi} \bra{\psi}} \,e^{-i s_\psi^{(N)} H} \right)^N \\
     & + O(s_\psi^{3/2} / \sqrt{N})
\label{eq:group_commutator}
\end{nalign}
for $s_\psi^{(N)} = \sqrt{|s_\psi| / N}$. This approximation allows implementing $U_{\psi}$ using the time-evolution operators $e^{\pm i s_\psi^{(N)} H}$ and the state-dependent phase operators $e^{\pm i s_\psi^{(N)} \ket{\psi} \bra{\psi}}$. We note that $H$ contains only two-body interactions,
\begin{equation}
    H = \sum_{i,j=1}^q (\delta_{ij} - \mathcal{D}_{ij}) \sigma^+_{V_i} \sigma^-_{V_j},
\end{equation}
where $\sigma^+= |1\rangle \langle 0|$ and $\sigma^-= |0\rangle \langle 1|$ and the notation $\mathcal{D}_{ij}$ represents the entry $(i,j)$ of the projector matrix $\mathcal{D} \in \mathbb{R}^{q \times q}$. The structure and sparsity of $H$ enable efficient compilations of the time evolution operators \cite{raeisi_quantum-circuit_2012, berry_exponential_2014, low_hamiltonian_2019, childs_hamiltonian_2012, gilyen_quantum_2019, berry_efficient_2007}. On the other hand, the phase operators can be realized via density matrix exponentiation methods \cite{lloyd_quantum_2014, rodriguez-grasa_quantum_2025}. Using an extra copy of $\ket{\psi}$, the protocol to implement the phase operators is expressed by the quantum channel
\begin{equation}
    \mathcal{T}[\sigma] = \mathrm{Tr}_2 [e^{i \Delta t \mathrm{SWAP} } (\sigma \otimes |\psi\rangle\langle\psi|) e^{-i \Delta t \mathrm{SWAP}} )],
\end{equation}
where $\mathrm{SWAP}$ denotes the swap operator between two registers. The error of this channel is bounded by
\begin{equation}
    \|\mathcal{T}[\sigma] - e^{i\Delta t |\psi\rangle\langle\psi| } \sigma e^{-i\Delta t |\psi\rangle\langle\psi| } \|_1 \leq \mathcal{O}(\Delta t^2)
\end{equation}
For $\Delta t = \pm s_\psi^{(N)}$, the protocol implements the phase operator $e^{ \pm i s_\psi^{(N)} \ket{\psi} \bra{\psi}}$ with error $\mathcal{O}(|s_\psi|/N)$. Combining these steps can facilitate a deterministic implementation of the projection $\ket{\psi} \mapsto U_\psi \ket{\psi} \approx \frac{\mathcal{D} \ket{\psi}}{\|\mathcal{D} \ket{\psi}\|}$, the key to fully coherent simulations of our QLBM algorithm. Even though this implementation allows for multi-timestep simulations without disruptive intermediate measurements, it would still require substantial quantum resources that are not available in the near term.

\section{Numerical Results}
\label{sec:numerics}
This section assesses the performance of the proposed QLBM algorithm applied to four problems: advection-diffusion under a Fourier-mode velocity, advection-diffusion of a Gaussian hill, flow of the Taylor-Green vortex, and flow around a cylinder. We compare the simulation results using the QLBM algorithm to those of the standard classical LBM (CLBM) algorithm and to the analytical solutions if computable. As required by the collision step of the QLBM algorithm, we fix unity relaxation rates $\tau = \tau_g = 1$, which correspond to kinematic viscosity $\nu = \frac{1}{6}$ for flow problems and diffusion coefficient $\kappa = \frac{1}{6}$ for advection-diffusion problems. The LBM simulations are initialized at the equilibrium distributions determined by the initial macroscopic fields, $f_i^{\text{eq}}(\mathbf{x},0) = f_i^{\text{eq}}(\rho(\mathbf{x},0), \mathbf{u}(\mathbf{x},0))$, whose expression is given by the second-order truncation in Eq.~\eqref{eq:equilibrium}. For each experiment, we run the simulations for $T=10000$ steps and keep track of the macroscopic density and velocity. The error of a macroscopic field $\mathbf{F}$ is quantified against the reference $\mathbf{F}_{\text{ref}}$ via the time-dependent relative $L_2$ error:
\begin{equation}
    \mathcal{E}_{\mathrm{Rel.\,} L_2}(t)  = \frac{\sum_{\mathbf{x}} \| \mathbf{F}(\mathbf{x},t) -  \mathbf{F}_{\text{ref}}(\mathbf{x},t) \|}{\sum_{\mathbf{x}} \|\mathbf{F}_{\text{ref}}(\mathbf{x},t) \|}
\end{equation}

\paragraph{ADE with Fourier-mode velocity}
We begin with a simple 1D advection-diffusion problem. The advection-diffusion equation (Eq.~\eqref{eq:advection_diffusion}) in the case of a periodic domain of length $L$, in which the advection velocity has the form $u_{\text{adv}}(x,t) = u_0 \cos(\lambda t)$, has an analytical solution. Note that this advection field is the solution to the Navier-Stokes equations with a pressure gradient $\nabla p \propto \sin(\lambda t)$. For example, when the initial concentration is $C(x,0) = C_0 + C_1 \cos(kx)$ for $k = 2\pi n/L$, $n\in \mathbb{Z}$, the time-dependent field is given by 
\begin{equation}
    C(x,t) = C_0 + C_1 e^{-\kappa k^2 t} \cos{(k(x-a(t)))},
\end{equation}
where $a(t) = \int_0^t u(s)\mathrm{d}s = \frac{u_0}{\lambda}\sin{\lambda t}$. 

Our experiment is conducted for the D1Q3 lattice with $L=256$, $u_0 = 0.1c_s$, $\lambda = 10^{-3}$, $C_0=1$, $C_1 = 0.5$, and $n=1$. The LBM distributions are initialized with $C(x,0)$ and $u(x,0) = u_0$. Fig.~\ref{fig:fourier} shows the solution obtained by the QLBM simulation and its accuracy. We find that the resulting field fully captures the time-dependent phase shift but slightly overshoots the analytical solution at local extrema, maintaining an error below $1\%$ throughout the simulation time.

\begin{figure}[!ht]
    \centering
    \begin{subfigure}[b]{\linewidth}
        \centering
        \includegraphics[width=\linewidth]{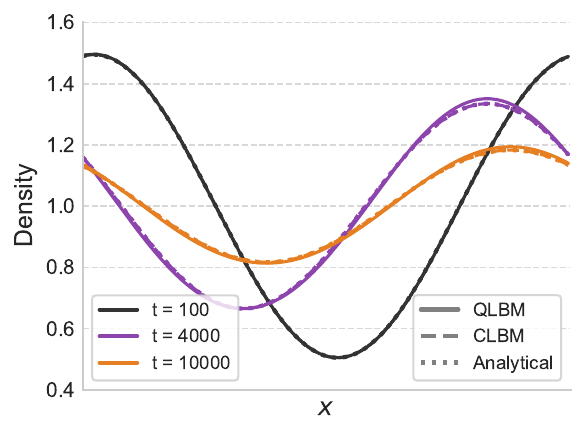}
        \caption{}
        \label{fig:fourier_plot}
    \end{subfigure}
    
    
    \begin{subfigure}[b]{1.0\linewidth}
        \centering
        \includegraphics[width=\linewidth]{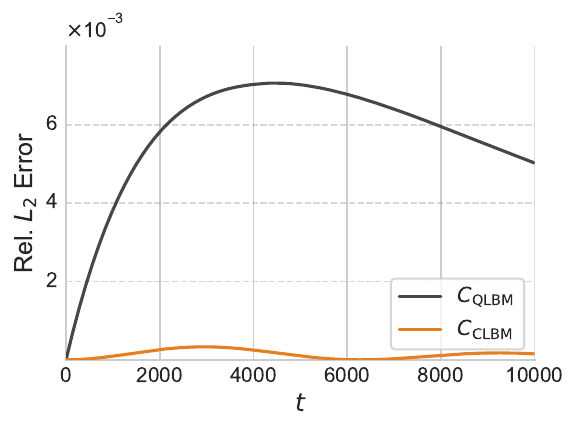}
        \caption{}
        \label{fig:fourier_error}
    \end{subfigure}
    
    \caption{\textbf{(a)} Concentration profiles $C(x,t)$ obtained with the QLBM at different times. \textbf{(b)} Relative $L_2$ error of the QLBM and CLBM solutions.}
    \label{fig:fourier}
\end{figure}

\paragraph{ADE of a Gaussian hill}
This test case examines the advection-diffusion process of a Gaussian hill in a 2D periodic domain under a constant uniform advection field $\mathbf{u}_{\text{adv}}$. Consider the isotropic initial profile $C(\mathbf{x},0) = C_0 \exp \left( - \frac{\|\mathbf{x} - \mathbf{x}_0\|^2}{2\sigma_0^2}\right)$, where $\mathbf{x}_0$ is the center of the initial Gaussian hill and $\sigma_0^2$ is the initial variance. The time evolution of the hill is given by 
\begin{equation}
    C(\mathbf{x},t) = \frac{C_0 \sigma_0^2}{\sigma_0^2 + \sigma_{\kappa}^2(t)} \exp \left( \frac{-\|\mathbf{x} - \mathbf{x}_0(t)\|^2}{2(\sigma_0^2 + \sigma_{\kappa}^2(t))} \right),
\end{equation}
where $\mathbf{x}_0(t) = \mathbf{x}_0 + \mathbf{u}_{\text{adv}}t$ and $\sigma_{\kappa}^2(t) = 2\kappa t$.

The experiment is conducted on a square domain of size $256 \times 256$ with $C_0 = 1$ and a constant advection field $\mathbf{u}_{\mathrm{adv}} = (0.3c_s,\,0.2c_s)^\top$. The LBM simulations employ the D2Q9 lattice and are initialized with the Gaussian profile $C(\mathbf{x},0)$ of width $\sigma_0 \in \{5,20,50\}$ and the prescribed advection field. Fig.~\ref{fig:gaussian} shows the evolution of the Gaussian hill and the simulation error. Similar to the 1D Fourier-mode velocity experiment, we find that the QLBM accurately captures the advective motion of the Gaussian hill, while its diffusion is underdamped. This behavior becomes more pronounced for smaller values of $\sigma_0$. The width $\sigma_0$ is related to the P\'eclet number ($\mathrm{Pe}$), which characterizes the relative strength of advection and diffusion:
\begin{equation}
\mathrm{Pe} = \frac{U l}{\kappa},
\end{equation}
where $U = \|\mathbf{u}_{\mathrm{adv}}\| \, \approx 0.208$ is the advection speed in lattice units, $l = \sigma_0$ is the characteristic length scale of the Gaussian hill, and the diffusivity $\kappa = 1/6$. This yields
$\mathrm{Pe} \approx 1.25\,\sigma_0$. Consequently, smaller $\sigma_0$ corresponds to a lower-$\mathrm{Pe}$, diffusion-dominated regime, in which the underdamped diffusion of the QLBM leads to larger errors.

\begin{figure}[!ht]
    \centering
    \begin{subfigure}[b]{0.75\linewidth}
        \centering
        \includegraphics[width=\linewidth]{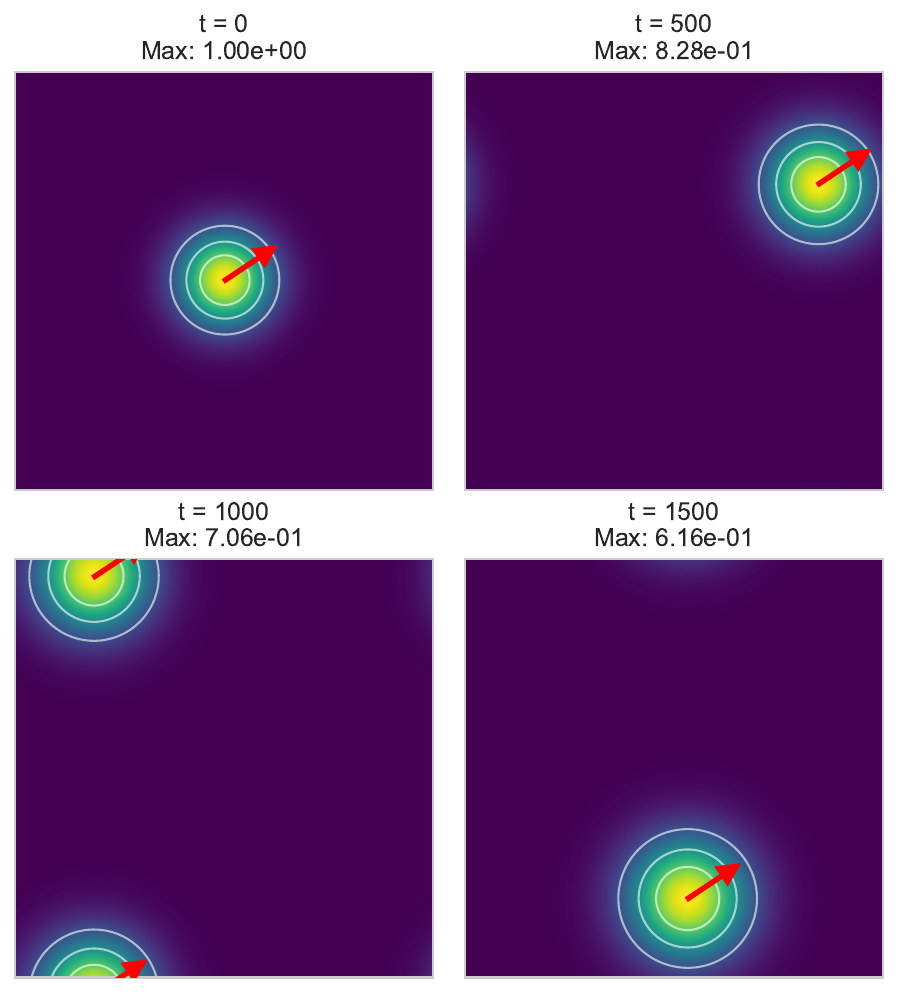}
        \label{fig:gaussian_field}
    \end{subfigure}
    
    \begin{subfigure}[b]{0.8\linewidth}
        \centering
        \includegraphics[width=\linewidth]{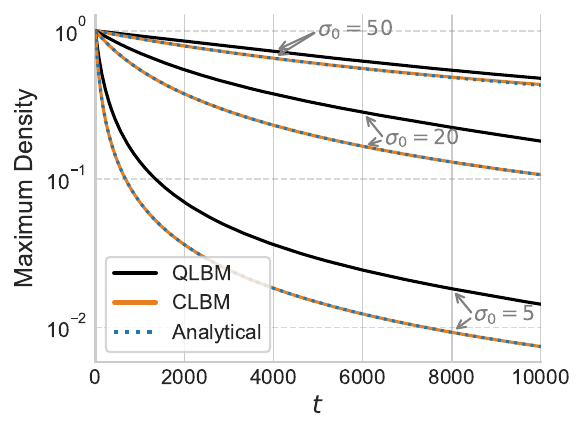}
        \label{fig:gaussian_decay}
    \end{subfigure}
    
    \begin{subfigure}[b]{0.8\linewidth}
        \centering
        \includegraphics[width=\linewidth]{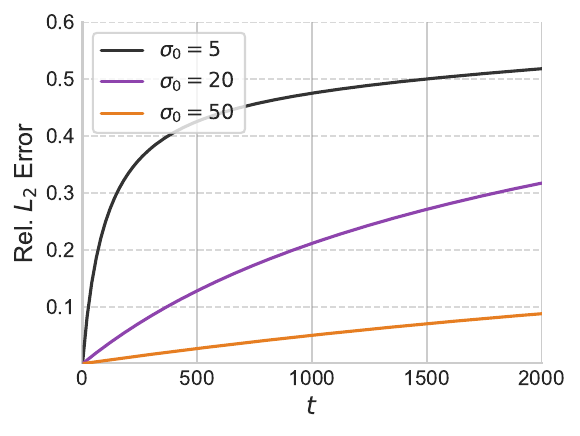}
        \label{fig:gaussian_error}
    \end{subfigure}
    
    \caption{\textbf{(top)} Time evolution of a Gaussian hill initialized with $\sigma_0 = 20$ in a periodic square domain; the red arrow indicates the advection direction $\mathbf{u}_{\text{adv}}$ and contours denote $75\%, 50\%$, and $25\%$ of the peak value. \textbf{(middle)} Diffusive decay of the peak density at the hill center. \textbf{(bottom)} Relative $L_2$ error of the QLBM simulations.}
    \label{fig:gaussian}
\end{figure}

\paragraph{Taylor-Green vortex decay}
This experiment tests how QLBM captures the temporal decay in a flow problem. The Taylor-Green flow is a periodic flow in a domain of size $L_x \times L_y$, specified by the velocity and pressure fields
\begin{equation}
    \mathbf{u}(\mathbf{x},t) = u_0 e^{-t/t_d} \begin{bmatrix} - \sqrt{k_y/k_x} \cos(k_x x) \sin(k_y y) \\ \sqrt{k_x/k_y} \sin(k_x x) \cos(k_y y)\end{bmatrix} 
\end{equation}
and
\begin{nalign}
    p(\mathbf{x},t) &= \rho_0 c_s^2 - \rho_0 \frac{u_0^2}{4} e^{-2t/t_d} \\
    & \; \; \times \left( \frac{k_y}{k_x} \cos(2k_x x) + \frac{k_x}{k_y} \cos(2k_y y) \right)
\end{nalign}
where $u_0$ is the initial velocity scale and $\rho_0$ is an arbitrary average density. Also, $(k_{x}, k_y) = (2\pi n_x/L_x, 2\pi n_y /L_y)$, $n_x, n_y \in \mathbb{Z}$, is the wave vector and $t_d = \frac{1}{\nu(k_x^2 + k_y^2)}$ the decay time scale. The pressure field is related to the density field via $p= \rho c_s^2$. 

The simulations are performed on a square $256 \times 256$ domain using the D2Q9 lattice, with $n_x = n_y = 1$, $\rho_0 = 1$, and $u_0 = 0.1c_s$. The initial state is specified by $\mathbf{u}(\mathbf{x},0)$ and $\rho(\mathbf{x},0)$. For the QLBM, a fixed reference velocity $\hat{\mathbf{u}} = (0,0)$ is used at every collision step. Fig.~\ref{fig:taylor_green} shows the velocity's magnitude and its error heatmap for the QLBM, as well as the relative error of the macroscopic fields. The velocity field in QLBM exhibits a large discrepancy early in the simulation, but the error decreases gradually. This is because the reference velocity $\hat{\mathbf{u}} = \mathbf{0}$ is only asymptotically exact; the velocity in early iterations is sinusoidal with magnitude $\sim u_0$. This suggests that choosing a good reference velocity is crucial to the accuracy of the QLBM simulation.

\begin{figure}[!ht]
    \centering
    \begin{subfigure}[b]{0.75\linewidth}
        \centering
        \includegraphics[width=\linewidth]{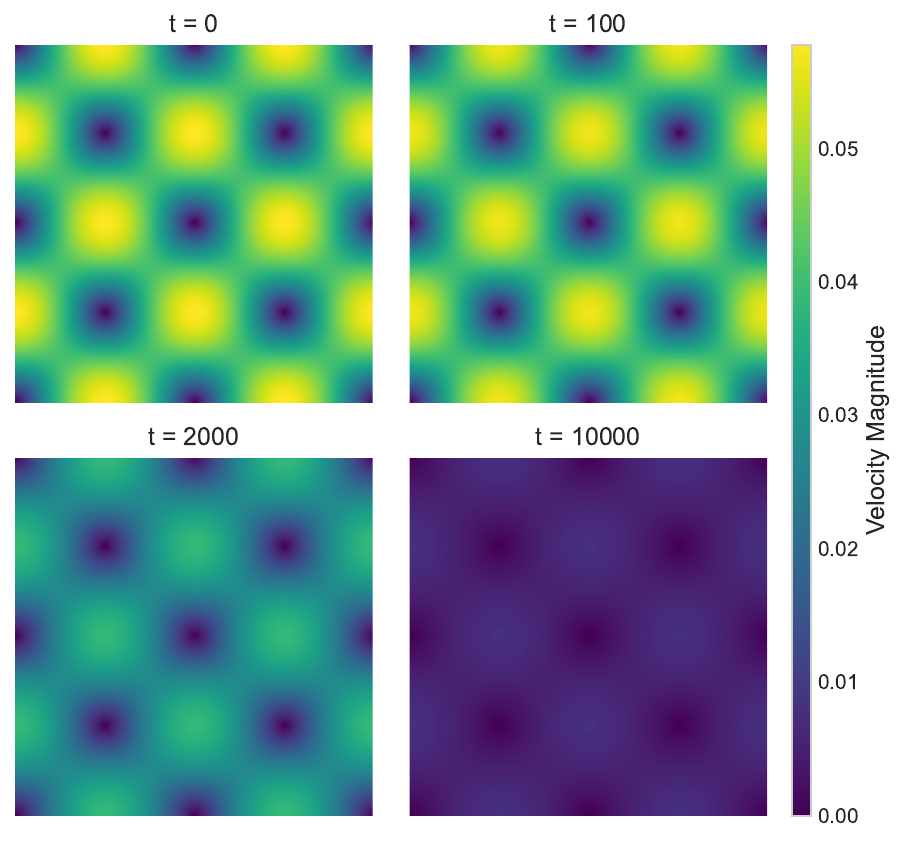}
        \label{fig:tg_field}
    \end{subfigure}

    
    \begin{subfigure}[b]{0.75\linewidth}
        \centering
        \includegraphics[width=\linewidth]{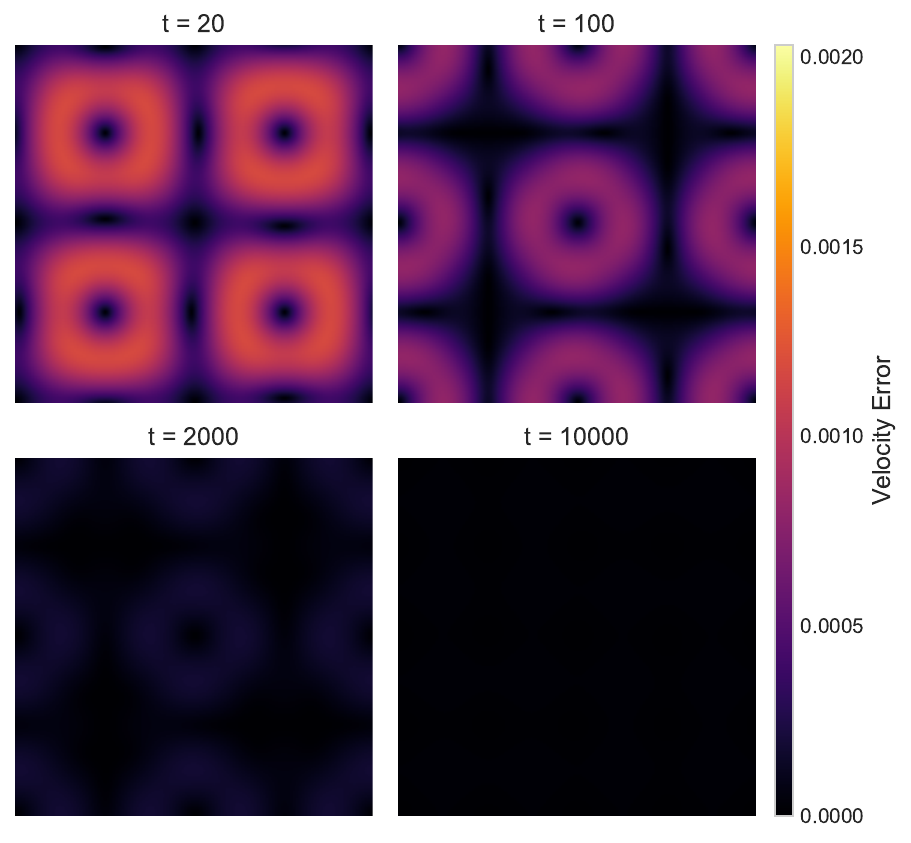}
        \label{fig:tg_heatmap}
    \end{subfigure}
    
    
    \begin{subfigure}[b]{0.8\linewidth}
        \centering
        \includegraphics[width=\linewidth]{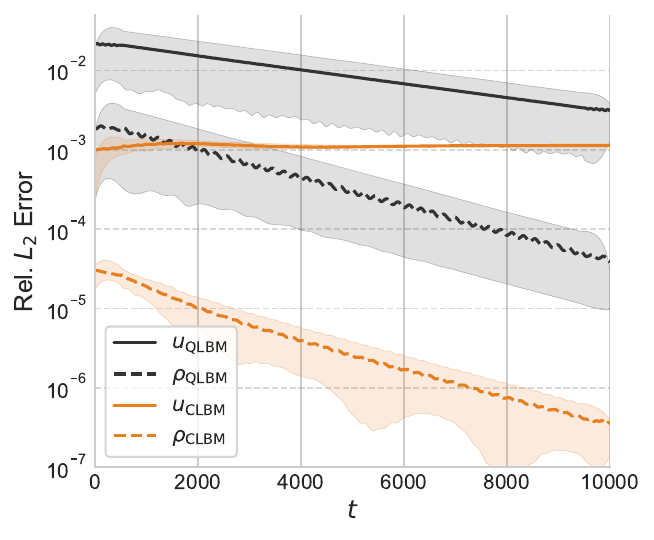}
        \label{fig:tg_error}
    \end{subfigure}
    
    \caption{\textbf{(top)} Temporal decay of the speed $\|\mathbf{u}\|$. \textbf{(middle)} Velocity error heatmaps. \textbf{(bottom)} Relative $L_2$ error of the density and velocity fields for the QLBM and CLBM solutions, shown as moving averages with shaded fluctuation envelopes.}
    \label{fig:taylor_green}
\end{figure}

\paragraph{Flow around a cylinder} This experiment demonstrates how the reference velocity affects a practical LBM simulation. The test case simulates the fluid flow in a 2D pipe with a circular obstacle. The geometry setting is a rectangular pipe of size $L_x \times L_y$ with periodic boundary and a solid circular cylinder of radius $L_y/8$ centered at $(L_y/2, L_y/2)$. The initial velocity applied to the left boundary is a parabolic profile $\mathbf{u}((0,y),t=0) = \frac{4u_0 y(L_y -y)}{L_y^2}$, where $u_0$ is the velocity scale.

We fix $L_x = 512, L_y = 128$, and $u_0 = 0.1c_s$. The LBM initial state is defined for the D2Q9 lattice with the uniform density $\rho(\mathbf{x},0) = 1$ and the parabolic profile $\mathbf{u}(\mathbf{x},0)$. We implement the no-slip boundary condition with half-way bounce-back at the boundary of the cylinder.

Figure~\ref{fig:cylinder_field} shows the velocity fields obtained with the CLBM
and QLBM. The QLBM simulation uses the reference velocity $\hat{\mathbf{u}} = \left(\tfrac{u_0}{3},\,0\right)$ and produces results comparable to the CLBM solution. This reference velocity is chosen \emph{a posteriori} as an approximation to the mean velocity of the CLBM flow at steady state. Compared to other choices, it yields lower error and higher stability, as shown in Fig.~\ref{fig:cylinder}. These results indicate that an accurate estimate of the flow velocity used as reference significantly improves the accuracy of the QLBM simulation.

\begin{figure}[!ht]
    \centering
    \includegraphics[width=\linewidth]{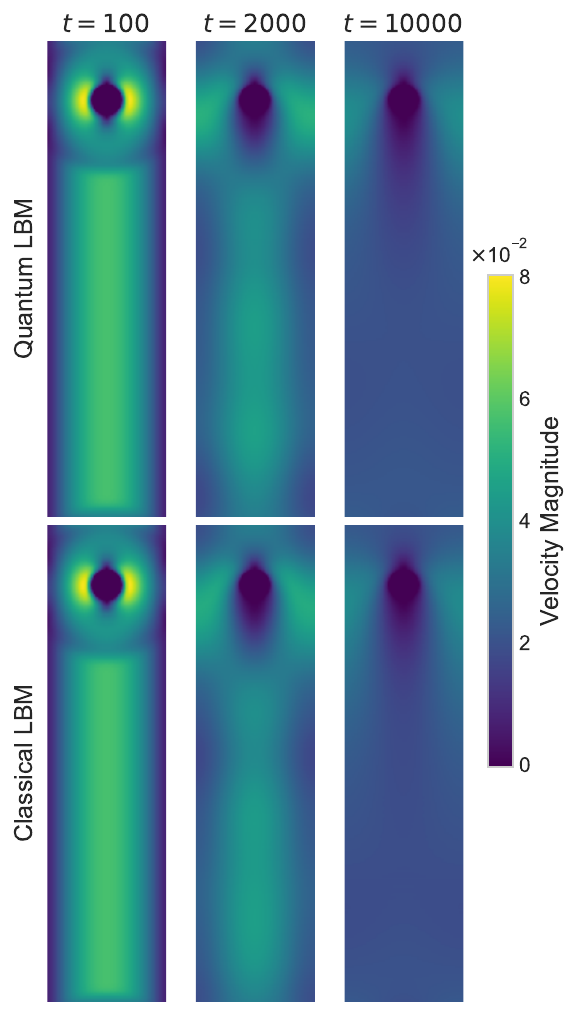}
    \caption{The speed field $\|\mathbf{u}\|$ from QLBM and CLBM at different times. Each plot is rotated by $90^\circ$ clockwise for visualization, so the original left boundary becomes the top edge.}
    \label{fig:cylinder_field}
\end{figure}

\begin{figure}[!ht]
    \centering
    \begin{subfigure}[b]{1.0\linewidth}
        \centering
        \includegraphics[width=\linewidth]{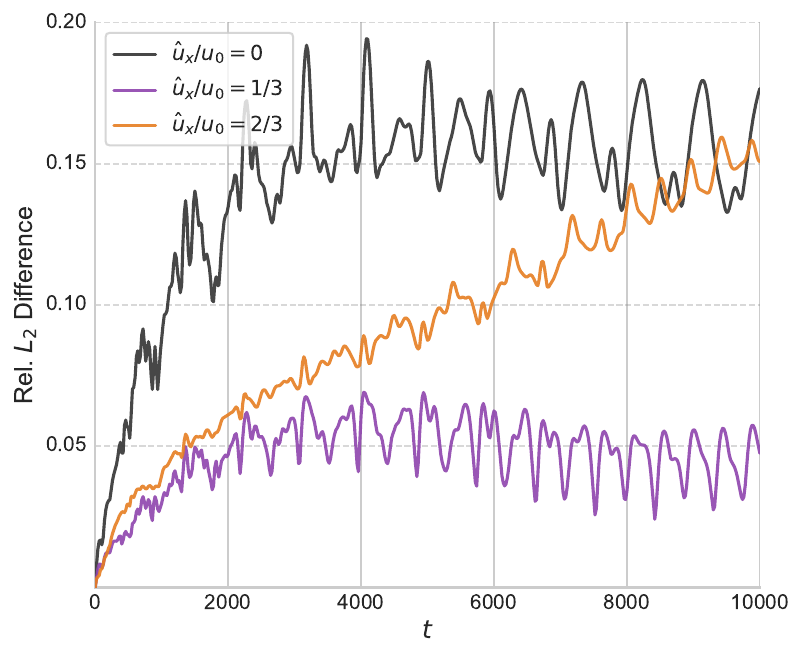}
        \caption{}
        \label{fig:cylinder_error}
    \end{subfigure}
    
    \begin{subfigure}[b]{1.0\linewidth}
        \centering
        \includegraphics[width=\linewidth]{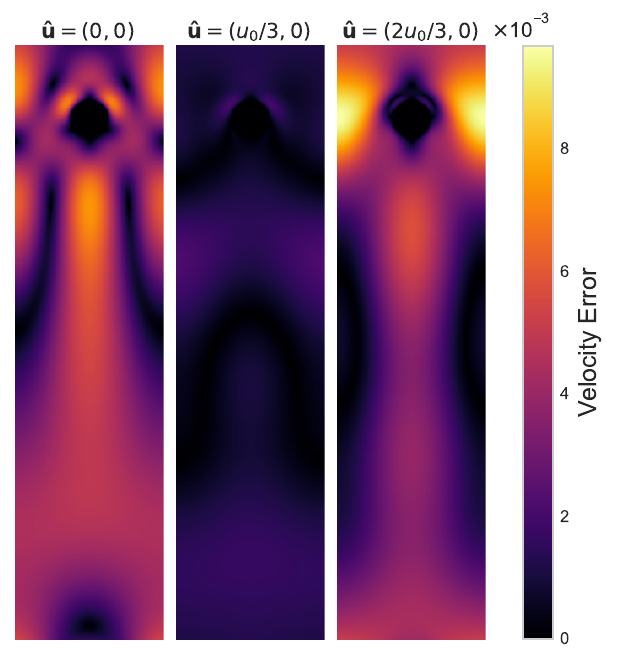}
        \caption{}
        \label{fig:cylinder_heatmap}
    \end{subfigure}
    
    \caption{\textbf{(a)} Relative $L_2$ difference between the QLBM solution and the CLBM solution for the horizontally oriented reference velocity $\hat{\mathbf{u}} = (\hat{u}_x, 0)$. \textbf{(b)} Error heatmaps for the velocity field at $t=10000$.}
    \label{fig:cylinder}
\end{figure}

\section{Conclusion and Outlook}
\label{sec:conclusion}
In this paper, we present a comprehensive quantum implementation framework for the Lattice Boltzmann method, based on one-hot encoding and a denoising-based collision operator. By formulating the collision step as an orthogonal projection onto the local linearization of the equilibrium manifold, we have demonstrated an alternative path to introducing the nonlinear physics of fluid flow within the constraints of linear quantum evolution. This approach effectively filters out non-equilibrium noise while preserving the essential nonlinearity in the Lattice Boltzmann equation.

Our theoretical framework, supported by Theorems \ref{thm:collision_error} and \ref{thm:denoising-error}, provides error bounds on the approximations inherent in this denoising approach. We established that the collision error is governed by the velocity misalignment and the intrinsic strain-rate tensor of the flow, while the manifold projection error is intrinsically linked to the curvature of the equilibrium manifold as quantified by its constant Hessian. These results suggest that the QLBM remains robust in regimes where the reference velocity tracks the macroscopic flow.

Numerical simulations across various benchmarks, including advection-diffusion problems, the Taylor-Green vortex, and flow around a cylinder, confirm the practical usability of the algorithm. The results indicate that the QLBM can capture advective motion and temporal decay, although it may exhibit underdamped diffusion in low-Péclet number regimes. A key takeaway from the numerical results is that the selection of an appropriate reference velocity is critical for ensuring accuracy and numerical stability.

Future research could focus on extending this framework to accommodate varying relaxation times $(\tau \neq 1)$ and developing automated, position-varying schemes for reference velocity updates to further enhance the accuracy of the denoising operator in complex flows. Finally, developing a deterministic multi-timestep evolution framework and integrating boundary conditions beyond the bounce-back schemes remain a challenge for applying QLBM to practical problems in computational fluid dynamics.

\paragraph{AI usage}
Large language models were used in the preparation of this manuscript. Google’s Gemini assisted with programming and visualization, while OpenAI’s ChatGPT supported manuscript organization and language editing. All scientific content, interpretations, and conclusions were produced by the authors, who take full responsibility for the work.

\paragraph{Acknowledgement} This work was supported by the Interdisciplinary Doctoral Program in Quantum Systems Integration funded by the BMW group. We thank C\u{a}lin Georgescu for insightful discussions.

\printbibliography
\bigskip

\onecolumn
\appendix

\section{Symmetry equivariance}
\label{appdx:equivariance}
Let $G$ be the symmetry group of the LBM lattice. A lattice symmetry $g \in G$ is an orthogonal transformation $R_g \in O(d)$ in physical space consisting of rotations and reflections that preserve the lattice structure. It acts on the discrete velocity set via a permutation $\sigma_{g} \in S_q$ such that $\mathbf{c}_{\sigma_g(i)} = R_g \mathbf{c}_i$. Consequently, its action on population vectors $\mathbf{f} = (f_1, \dots, f_q)$ is defined by the permutation matrix $P_g \in \mathbb{R}^{q \times q}$, with $(P_g \mathbf{f})_i = f_{\sigma^{-1}_g(i)}$. Since the collision operator $\Omega = \mathcal{D}(\hat{\mathbf{u}})$ depends on the reference velocity, this reference must be transformed accordingly when the collision operator is applied to a permuted population vector. That is, $\Omega[g \cdot \mathbf{f}] = \mathcal{D}(R_g \hat{\mathbf{u}}) P_g  \mathbf{f}$. The symmetry equivariance condition (Eq.~\eqref{eq:collision_symm_equiv}) therefore takes the form $\mathcal{D}(R_g \hat{\mathbf{u}}) P_g \mathbf{f} = P_g  \mathcal{D}(\hat{\mathbf{u}}) \mathbf{f}$ for every $\mathbf{f} \in \mathbb{R}^q$. We will prove the equivalent operator identity, $\mathcal{D}(R_g \hat{\mathbf{u}}) = P_g  \mathcal{D}(\hat{\mathbf{u}}) P_g ^\top$.

The first step is to prove $P_g \mathbf{h}(\mathbf{u}) = \mathbf{h}(R_g \mathbf{u})$. By definition of the permutation matrix $P_g$, $(P_g \mathbf{h}(\mathbf{u}))_i = h_{\sigma^{-1}_g(i)} (\mathbf{u})$. The equilibrium amplitude $\mathbf{h}(\mathbf{u})$ is constructed from the discrete velocities and weights such that each component depends on $\mathbf{u}$ through the scalar quantities $\mathbf{c}_i \cdot \mathbf{u}$ and $\|\mathbf{u}\|$. Since the lattice weights are invariant under the symmetry, and $R_g$ is orthogonal, we have
\begin{equation}
    \mathbf{c}_{\sigma_g^{-1}(i)} \cdot \mathbf{u} = (R_g^\top \mathbf{c}_i) \cdot \mathbf{u} = \mathbf{c}_i \cdot (R_g \mathbf{u}), \quad \|\mathbf{u}\| = \|R_g \mathbf{u}\|
\end{equation}
Then, $(P_g \mathbf{h}(\mathbf{u}))_i = h_{\sigma^{-1}_g(i)} (\mathbf{u}) = h_i(R_g \mathbf{u})$ holds for every component $i$. Hence $P_g \mathbf{h}(\mathbf{u}) = \mathbf{h}(R_g \mathbf{u})$. We next evaluate the transformation of the Jacobian of the equilibrium amplitudes. Recall that the denoising operator is defined as
\begin{equation}
\mathcal{D}(\hat{\mathbf{u}})
=
\bar J(\hat{\mathbf{u}})
\big(\bar J(\hat{\mathbf{u}})^\top \bar J(\hat{\mathbf{u}})\big)^{-1}
\bar J(\hat{\mathbf{u}})^\top,
\end{equation}
where $\bar J(\hat{\mathbf{u}}) = [\mathbf{h}(\hat{\mathbf{u}}), \partial_1 \mathbf{h}(\hat{\mathbf{u}}), \dots, \partial_d \mathbf{h}(\hat{\mathbf{u}})] \in\mathbb{R}^{q\times(d+1)}$ denotes the scaled Jacobian of $\mathbf{g}(\rho, \mathbf{u}) = \sqrt{\rho}\mathbf{h}(\mathbf{u})$ evaluated at the reference velocity $\hat{\mathbf{u}}$. Differentiating the identity $P_g \mathbf{h}(\mathbf{u}) = \mathbf{h}(R_g \mathbf{u})$ and evaluating at $\mathbf{u}=\hat{\mathbf{u}}$ yields $P_g \, \nabla \mathbf{h}(\hat{\mathbf{u}}) = \nabla \mathbf{h}(R_g \hat{\mathbf{u}}) \, R_g$. Combining the original identity and its derivatives, we obtain
\begin{equation}
P_g \bar J(\hat{\mathbf{u}})
=
\bar J(R_g \hat{\mathbf{u}})\, Q_g,
\end{equation}
where $Q_g = \begin{pmatrix} 1 & 0 \\ 0 & R_g \end{pmatrix}$ is an orthogonal operator. We now compare the denoising operators at $\hat{\mathbf{u}}$ and $R_g\hat{\mathbf{u}}$. Using the relation above, we compute
\begin{nalign}
P_g \mathcal{D}(\hat{\mathbf{u}}) P_g^\top
&=
P_g \bar J(\hat{\mathbf{u}})
\big(\bar J(\hat{\mathbf{u}})^\top \bar J(\hat{\mathbf{u}})\big)^{-1}
\bar J(\hat{\mathbf{u}})^\top P_g^\top \\
&=
\bar J(R_g \hat{\mathbf{u}})\, Q_g
\big(Q_g^\top \bar J(R_g \hat{\mathbf{u}})^\top
\bar J(R_g \hat{\mathbf{u}})\, Q_g\big)^{-1}
Q_g^\top \bar J(R_g \hat{\mathbf{u}})^\top \\
&= 
\bar J(R_g \hat{\mathbf{u}})\, Q_g \,
Q_g^{-1} \big(\bar J(R_g \hat{\mathbf{u}})^\top
\bar J(R_g \hat{\mathbf{u}}) \big)^{-1} Q_g^{-\top}
Q_g^\top \bar J(R_g \hat{\mathbf{u}})^\top \\
&= \bar J(R_g \hat{\mathbf{u}}) \big(\bar J(R_g \hat{\mathbf{u}})^\top
\bar J(R_g \hat{\mathbf{u}}) \big)^{-1} J(R_g \hat{\mathbf{u}})^\top = \mathcal{D}(R_g \hat{\mathbf{u}}).
\end{nalign}
This completes the proof for the equivariance with respect to the symmetry group of the lattice, $\mathcal{D}(R_g \hat{\mathbf{u}}) P_g = P_g \mathcal{D}(\hat{\mathbf{u}})$ for any $g \in G$.

\section{Error analysis}
\label{appdx:error_analysis}
\subsection*{Collision error}
\label{appdx:collision_error}
This part estimates the error of using the projection onto the linearized manifold as the collision operator. Let $\hat{V}$ be the tangent space at a reference velocity $\hat{\mathbf{u}}$ of the equilibrium manifold $\mathcal{M} \subset \mathbb{R}^q$, which is defined by $\mathbf{g}(\rho, \mathbf{u}) = \sqrt{\rho} \, \mathbf{h}(\mathbf{u})$ (Eq. \eqref{eq:velocity_component}). Note that this tangent space is invariant with respect to density $\rho$ and contains the vectors $\mathbf{g}(\rho, \hat{\mathbf{u}})$. Let the post-streaming amplitudes at a lattice node be $\tilde{\mathbf{g}} = \mathbf{g}(\rho^{\text{str}}, \mathbf{u}^{\text{str}}) + \boldsymbol{\xi}$, where $\boldsymbol{\xi}$ represents the non-equilibrium deviation. We evaluate the error $\mathcal{E} = \| \Pi_{\hat{V}} \tilde{\mathbf{g}} - \mathbf{g}(\rho^{\text{str}}, \mathbf{u}^{\text{str}}) \|$ of approximating the true equilibrium state via projection onto $\hat{V}$. Using $\mathbf{g} = \Pi_{\hat{V}} \mathbf{g} + \Pi_{\hat{V}^\perp}\mathbf{g}$, where $\hat{V}^\perp$ is the orthogonal complement to $\hat{V}$, we can rewrite the error as
\begin{equation}
    \mathcal{E} = \| - \Pi_{\hat{V}^\perp} \mathbf{g}(\rho^{\text{str}}, \mathbf{u}^{\text{str}}) + \Pi_{\hat{V}} \boldsymbol{\xi}\|
\end{equation}
By the triangle inequality, $\mathcal{E} \leq \| \Pi_{\hat{V}^\perp} \mathbf{g}(\rho^{\text{str}}, \mathbf{u}^{\text{str}}) \| + \| \Pi_{\hat{V}} \boldsymbol{\xi}\|$. To bound the first error term, we expand $\mathbf{g}(\rho^{\text{str}}, \mathbf{u}^{\text{str}})$ about $\hat{\mathbf{u}}$:
\begin{equation}
    \mathbf{g} = \sqrt{\rho^{\text{str}}} \left( \mathbf{h}(\hat{\mathbf{u}}) + J \Delta \mathbf{u} + \frac{\Delta\mathbf{u}^\top \mathbf{H} \Delta\mathbf{u}}{2} \right)
\end{equation}
where $\Delta \mathbf{u} = \mathbf{u}^{\text{str}} - \hat{\mathbf{u}}$, and $J$ and $\mathbf{H}$ are the Jacobian and Hessian of $\mathbf{h}$ at $\hat{\mathbf{u}}$, respectively. Since $\mathbf{h}(\mathbf{u})$ is quadratic, the expansion stops at the second order and the Hessian is independent of $\hat{\mathbf{u}}$. By definition, $\mathbf{h}(\hat{\mathbf{u}})$ and the columns of $J$ span $\hat{V}$. The projection onto $\hat{V}^\perp$ annihilates these terms, leaving only the quadratic part,
\begin{nalign}
    \| \Pi_{\hat{V}^\perp} \mathbf{g}(\rho^{\text{str}}, \mathbf{u}^{\text{str}}) \| 
    &= \frac{\sqrt{\rho^{\text{str}}}}{2} \| \Pi_{\hat{V}^\perp} \Delta\mathbf{u}^\top \mathbf{H} \Delta\mathbf{u}\| \\
    &\leq \frac{\sqrt{\rho^{\text{str}}}}{2} \| \Delta\mathbf{u}^\top \mathbf{H} \Delta\mathbf{u}\| \\
    &\leq \frac{\sqrt{\rho^{\text{str}}}}{2} \|\mathbf{H}\| \|\Delta \mathbf{u}\|^2
\end{nalign}

To understand the second error term, we relate the deviation $\boldsymbol{\xi}$ to standard quantities in LBM. A distribution $\tilde{f}_i $ can be decomposed into the equilibrium part $f_i^{\text{eq}}$ and the non-equilibrium part $f_i^{\text{neq}}$ via $\tilde{f}_i = f_i^{\text{eq}} + f_i^{\text{neq}}$. Assuming $f_i^{\text{neq}} \ll f_i^{\text{eq}}$, we approximate $\tilde{g}_i = \sqrt{f_i^{\text{eq}} + f_i^{\text{neq}}} \approx \sqrt{f_i^{\text{eq}}} \left( 1 + \frac{f_i^{\text{neq}}}{2f_i^{\text{eq}}}\right)$. The deviation $\boldsymbol{\xi}$ therefore has the first-order approximation $\xi_i \approx \frac{f_i^{\text{neq}}}{2\sqrt{f_i^{\text{eq}}(\rho^{\text{str}}, \mathbf{u}^{\text{str}})}}$. On the other hand, it can be shown that applying the Chapman-Enskog expansion to LBM leads to the approximation \cite{kruger_lattice_2017}
\begin{equation}
    f_i^{\text{neq}} \approx - \rho \frac{w_i \tau}{c_s^2} \sum_{k,l=1}^d Q_{ikl} \, \partial_k u_l
\end{equation}
where $Q_{ikl} = c_{ik} c_{il} - \delta_{kl} c_s^2$ is the velocity tensor. The velocity gradient tensor $\partial_k u_l$ has a unique decomposition into a symmetric part and an antisymmetric part
\begin{equation}
    \partial_k u_l = \frac{1}{2} (\partial_k u_l + \partial_l u_k) + \frac{1}{2} (\partial_k u_l - \partial_l u_k),
\end{equation}
where $(\partial_k u_l + \partial_l u_k)/2 =: S_{kl}$ is called the strain-rate tensor and $(\partial_k u_l - \partial_l u_k)/2 =: \Omega_{kl}$ the vorticity tensor. Since $Q_{ikl}$ is symmetric with respect to the indices $k$ and $l$, the antisymmetric vorticity tensor does not contribute to the non-equilibrium for $\sum_{k,l} Q_{ikl} \, \Omega_{kl} = 0$. So, the non-equilibrium part is related to the strain rate via $f_i^{\text{neq}} \approx - \rho \frac{w_i \tau}{c_s^2} \sum_{k,l=1}^d Q_{ikl} \, S_{kl}$. Then
\begin{equation}
    \xi_i \approx -\frac{\tau}{2c_s^2}  \sqrt{\rho w_i} \sum_{k,l} Q_{ikl} \, S_{kl}
\end{equation}

We use isotropy conditions to evaluate the fourth moment
\begin{equation}
    \sum_i w_i Q_{ikl} Q_{irs} = 2c_s^4 \left( \delta_{kr} \, \delta_{ls} + \delta_{ks} \, \delta_{lr} \right)
\end{equation}
Substitute this into the Euclidean norm $\|\boldsymbol{\xi}\|$:
\begin{nalign}
    \|\boldsymbol{\xi}\|^2 &
    \approx \frac{\tau^2 \rho}{4c_s^4} \sum_{i} \sum_{k,l,r,s}  w_i (Q_{ikl} \, S_{kl}) (Q_{irs} \, S_{rs}) \\
    &= \frac{\tau^2 \rho}{4} \sum_{k,l} (S_{kl} S_{kl} + S_{kl} S_{lk}) \\
    &= \frac{\tau^2 \rho}{2} \sum_{k,l} S_{kl} S_{kl} = \frac{\tau^2 \rho}{2} \|S\|_{\mathrm{F}}^2,
\end{nalign}
where $\|S\|_{\mathrm{F}}$ is the Frobenius norm of the strain-rate tensor. Since $\tau = 1$ in our setting, the second error term is bounded by $\| \Pi_{\hat{V}} \boldsymbol{\xi}\| \leq \|\boldsymbol{\xi}\| \approx \sqrt{\frac{\rho^{\text{str}}}{2}} \|S(\mathbf{u}^{\text{str}})\|_{\mathrm{F}}$. Combining the two sources of error, we have
\begin{equation}
    \mathcal{E} \leq \frac{\sqrt{\rho^{\text{str}}}}{2} \left( \|\mathbf{H}\| \|\Delta \mathbf{u}\|^2 + \sqrt{2} \|S(\mathbf{u}^{\text{str}})\|_{\mathrm{F}} \right) 
\end{equation}
The first error term is a misalignment error that represents how far the reference velocity is from the post-streaming velocity. The second term predicts an intrinsic error in the regions where the flow has intense velocity gradients.

\subsection*{Manifold projection error}
\label{appdx:projection_error}
This part evaluates the error of the denoising operator as an approximate projection onto the manifold. We reuse the notations defined in the previous analysis of the collision error. The manifold projection error is measured by
\begin{equation}
    \mathcal{E}^\perp = \text{dist}(\Pi_{\hat{V}} \tilde{\mathbf{g}}, \mathcal{M}) = \min_{\rho,\mathbf{u}} \| \Pi_{\hat{V}} \tilde{\mathbf{g}} - \sqrt{\rho} \mathbf{h}(\mathbf{u})\|
\end{equation}
The distance square decomposes into a tangential term, $\| \Pi_{\hat{V}} (\Pi_{\hat{V}} \tilde{\mathbf{g}} - \sqrt{\rho}\mathbf{h}(\mathbf{u}))\|^2 = \| \Pi_{\hat{V}} (\tilde{\mathbf{g}} - \sqrt{\rho}\mathbf{h}(\mathbf{u}))\|^2$,  and a normal term, $\| \Pi_{\hat{V}^\perp} (\Pi_{\hat{V}} \tilde{\mathbf{g}} - \sqrt{\rho}\mathbf{h}(\mathbf{u}))\|^2 = \| \Pi_{\hat{V}^\perp} \sqrt{\rho} \mathbf{h}(\mathbf{u})\|^2$

We demonstrate that the first term vanishes at specific macroscopic values. The main tool is the Newton-Kantorovich theorem \cite{ciarlet_newtonkantorovich_2012}, which provides sufficient conditions for the existence and uniqueness of a solution to the equation $\Pi_{\hat{V}} \sqrt{\rho} \mathbf{h}(\mathbf{u}) = \Pi_{\hat{V}}  \tilde{\mathbf{g}}$.

Let $\mathcal{F}(\rho, \mathbf{u}) = \Pi_{\hat{V}} \sqrt{\rho} \mathbf{h}(\mathbf{u})$ be the mapping from the macroscopic quantities to the tangent space. We seek a root $(\rho^\dagger, \mathbf{u}^\dagger)$ for the equation $\mathcal{F}(\rho, \mathbf{u}) - \Pi_{\hat{V}} \tilde{\mathbf{g}} = 0$ following the Newton's method starting from the initial guess $(\rho_0, \mathbf{u}_0) = (\hat\rho, \hat{\mathbf{u}})$. First, we evaluate the Jacobian $\hat{D}_{\mathcal{F}} = D_{\mathcal{F}}(\hat\rho, \hat{\mathbf{u}})$:
\begin{nalign}
    \partial_{\rho} \mathcal{F} &= \Pi_{\hat{V}} \frac{\mathbf{h}(\hat{\mathbf{u}})}{ 2\sqrt{\hat \rho}} = \frac{\mathbf{h}(\hat{\mathbf{u}})}{ 2\sqrt{\hat \rho}} \\
    \partial_{u_k} \mathcal{F} &= \Pi_{\hat{V}} \sqrt{\hat{\rho}} \, \partial_k \mathbf{h}(\hat{\mathbf{u}}) = \sqrt{\hat \rho} \, \partial_k \mathbf{h}(\hat{\mathbf{u}})
\end{nalign}
The identities hold because $\hat{V} = \mathrm{span} \{\mathbf{h}(\hat{\mathbf{u}}), \partial_1 \mathbf{h}(\hat{\mathbf{u}}), \dots, \partial_d \mathbf{h}(\hat{\mathbf{u}})\}$. Since $\mathrm{dim}(\hat{V}) = d+1$, the Jacobian $\hat{D}_{\mathcal{F}}$, when restricted to $\hat{V}$, is non-singular. With an abuse of notation, we define $\hat{D}_{\mathcal{F}}^{-1}: \hat{V} \rightarrow \mathbb{R}^{d+1}$ as the inverse of $\hat{D}_{\mathcal{F}}$ restricted to $\hat{V}$. The inverse can therefore act on $\mathcal{F}(\rho, \mathbf{u}) - \Pi_{\hat{V}} \tilde{\mathbf{g}} \in \hat{V}$.

We apply the Newton-Kantorovich theorem. Let $\mathbf{z}_0 := \hat{D}_{\mathcal{F}}^{-1} (\mathcal{F}(\rho_0, \mathbf{u}_0) - \Pi_{\hat{V}} \tilde{\mathbf{g}})$ be the first Newton update, i.e., the next point is $(\rho_1, \mathbf{u}_1) = ( \rho_0, \mathbf{u}_0) - \mathbf{z}_0$. Assume the open ball $B \equiv B((\rho_1, \mathbf{u}_1), r)$ with radius $r=\|\mathbf{z}_0\|$ is contained in $\mathbb{R}^+ \times \mathbb{R}^d$.

Let $\beta = \| \hat{D}_{\mathcal{F}}^{-1}\| < \infty$ be the operator norm. Since $\hat{D}_{\mathcal{F}}$ has full column rank, $\beta$ is equal to the operator norm of the pseudoinverse $\hat{D}_{\mathcal{F}}^+$ and related to the smallest singular value $\sigma_{\text{min}} (\hat{D}_{\mathcal{F}}) > 0$ by
$\beta = \|\hat{D}_{\mathcal{F}}^+\| = \frac{1}{\sigma_{\text{min}} (\hat{D}_{\mathcal{F}})}$.

Let $\mu$ be the Lipschitz constant of $D_{\mathcal{F}}$ over $B$. It can be bounded in terms of the second derivatives as $\mu \leq \sup_{B} \| D^2_\mathcal{F}(\rho, \mathbf{u})\|_2$. Here, $D^2_\mathcal{F}(\rho, \mathbf{u})$ is a $q \times (d+1) \times (d+1)$ tensor that can be written in  block form as
\begin{equation}
    D^2_\mathcal{F}(\rho, \mathbf{u}) = 
    \begin{bmatrix}
        - \frac{1}{4\rho^{3/2}}\mathbf{h}(\mathbf{u}) & \frac{1}{2\sqrt{\rho}} \nabla \mathbf{h}(\mathbf{u}) \\
        \frac{1}{2\sqrt{\rho}} \nabla \mathbf{h}(\mathbf{u}) & \sqrt{\rho} \nabla^2 \mathbf{h}(\mathbf{u})
    \end{bmatrix}
\end{equation}
Since $\mathbf{h}$ is quadratic, $\nabla \mathbf{h}(\mathbf{u})$ is a first-order function in $\mathbf{u} - \mathbf{u}_1$ and $\nabla^2 \mathbf{h}(\mathbf{u}) = \mathbf{H}$ a constant Hessian tensor. These allow one to estimate an upper bound of $\mu$.

The key assumption for applying the Newton-Kantorovich theorem is $2\beta \mu r \leq 1$. According to the theorem, if this condition holds, then the Newton iteration converges to a solution $(\rho^\dagger, \mathbf{u}^\dagger)$ in the closed ball $\bar{B}$ of the equation $\Pi_{\hat{V}} \sqrt{\rho} \mathbf{h}(\mathbf{u}) = \Pi_{\hat{V}}  \tilde{\mathbf{g}}$. This implies the tangential error term vanishes at this solution.

We will examine the normal error term evaluated at $(\rho^\dagger, \mathbf{u}^\dagger)$. We write $\mathbf{h}(\mathbf{u}^\dagger) = \mathbf{h}(\hat{\mathbf{u}}) + J \Delta \mathbf{u}^\dagger + \frac{1}{2} (\Delta \mathbf{u}^\dagger)^\top \mathbf{H} \Delta \mathbf{u}^\dagger$, where $\Delta \mathbf{u}^\dagger = (\mathbf{u}^\dagger - \hat{\mathbf{u}})$. Since the first two terms lie in $\hat{V}$, we have
\begin{nalign}
    \mathcal{E}^\perp \leq \| \Pi_{\hat{V}^\perp} \sqrt{\rho^\dagger} \mathbf{h}(\mathbf{u}^\dagger) \| 
    &= \frac{\sqrt{\rho^\dagger}}{2} \| \Pi_{\hat{V}^\perp} (\Delta\mathbf{u}^\dagger)^\top \mathbf{H} \Delta\mathbf{u}^\dagger\| \\
    &\leq \frac{\sqrt{\rho^\dagger}}{2} \|\mathbf{H}\| \|\Delta \mathbf{u}^\dagger \|^2 \\
    &\leq 2\sqrt{\hat \rho + 2r} \|\mathbf{H}\|  r^2.
\end{nalign}
The final inequality holds because both $(\hat \rho, \hat{\mathbf{u}})$ and $(\rho^\dagger, \mathbf{u}^\dagger)$ lie within the ball $\bar{B}$ of radius $r$. Furthermore, this radius is bounded by:
\begin{nalign}
r 
&= \|\hat{D}_{\mathcal{F}}^{-1} \Pi_{\hat{V}} (\sqrt{\hat\rho} \mathbf{h}(\hat{\mathbf{u}}) -  \tilde{\mathbf{g}}) \| \\
&\leq \|\hat{D}_{\mathcal{F}}^{-1}\| \|\Pi_{\hat{V}}\| \Delta, & \Delta = \| \sqrt{\hat\rho} \mathbf{h}(\hat{\mathbf{u}}) -  \tilde{\mathbf{g}} \| \\
&\leq \beta \Delta.
\end{nalign}
Consequently, $\mathcal{E}^\perp \leq 2 \beta^2 \Delta^2 \|\mathbf{H}\| \sqrt{\hat \rho + 2\beta \Delta}$. This bound can be tightened by choosing $\sqrt{\hat{\rho}} = \mathbf{h}(\hat{\mathbf{u}})^\top \tilde{\mathbf{g}} / \|\mathbf{h}(\hat{\mathbf{u}})\|^2$ to minimize $\Delta$. We can then express the residual as $\Delta = \|\tilde{\mathbf{g}}\| \left|\sin\theta\right|$, where the angle $\theta$ is defined by $\cos\theta = \frac{\mathbf{h}(\hat{\mathbf{u}})^\top \tilde{\mathbf{g}}}{\|\mathbf{h}(\hat{\mathbf{u}})\| \|\tilde{\mathbf{g}}\|}$.


\end{document}